\documentclass[12pt]{article}
\usepackage{epsfig,latexsym,amsfonts,amsmath,amsthm,amssymb,amsbsy,multirow,slashed,wasysym,textcomp,dsfont,comment,mathtools,cancel,slashed,cite,datetime,appendix,BOONDOX-calo}
\usepackage{tocloft}
\usepackage{bbold}

\usepackage{graphicx}
\usepackage{float}
\usepackage{tikz}
\usetikzlibrary{backgrounds}
\usetikzlibrary{patterns}
\usetikzlibrary{arrows.meta}
\usetikzlibrary{shapes,snakes}
\tikzstyle{bag} = [align=center]
\usetikzlibrary{decorations.pathmorphing}
\usetikzlibrary{decorations.markings}
\usetikzlibrary{patterns, arrows.meta}
\usetikzlibrary{fadings}
\usetikzlibrary{shapes.misc}
\newtheorem{proposition}{Proposition}
\newtheorem{theorem}{Theorem}
\newtheorem{definition}{Definition}
\newtheorem{lemma}{Lemma}
\tikzset{cross/.style={cross out, draw=blue, minimum size=2*(#1-\pgflinewidth), inner sep=0pt, outer sep=0pt},
cross/.default={1.7pt}}
\definecolor{newred}{rgb}{0.65, 0.16, 0.16}
\definecolor{brg}{rgb}{0.0, 0.26, 0.45}

\usepackage{xspace}
\usepackage{hyperref}
\usepackage{subcaption}
\usepackage{ytableau}
\usepackage[vcentermath]{youngtab}
 \newcommand{\badat}{\begin{alignedat}}
 \newcommand{\eadat}{\end{alignedat}}
 
 \def\be{\begin{equation}}
\def\ee{\end{equation}}

\def\p{\partial}

\usepackage{xcolor}

\newcommand{\pink}[1]{\textcolor{\pink}{#1}}

\definecolor{dblue}{rgb}{0.2,0.50,0.80}

 \definecolor{dred}{rgb}{0.65,0.10,0.20} 
\definecolor{dblue}{rgb}{0.2,0.50,0.80}

\def\Acal{\mathcal{A}}

\def\Scal{\mathcal{S}}

\def\bz{{\bar z}}

\usepackage{comment}
\usepackage{todonotes}

\numberwithin{equation}{section} % equation numbers follow sections

\begin{document}

 \begin{titlepage}
  \thispagestyle{empty}
  \begin{flushright}
  CPHT-RR078.122023
  \end{flushright}
  
  \bigskip \bigskip
   
  \begin{center}

        \baselineskip=13pt {\Huge \scshape{
       Distributional Celestial Amplitudes
        %\\ \vspace{0.5cm} 
        }}
  
      \vskip1cm 

   \centerline{\large Majdouline Borji$^{*\dag}$, Yorgo Pano$^*$
   }

\bigskip \bigskip
 
\centerline{\em $^*$CPHT, CNRS, Ecole Polytechnique, IP Paris, F-91128 Palaiseau, France}
\centerline{\em $^\dag$Max-Planck Institute for Mathematics in the Sciences
Inselstr. 22, 04103 Leipzig, Germany}

\bigskip \bigskip
  
\end{center}

\begin{abstract}
  \noindent 
Scattering amplitudes are tempered distributions, which are defined through their action on functions in the Schwartz space $\mathcal{S}(\mathbb{R})$ by duality. For massless particles, their conformal properties become manifest when considering their Mellin transform. Therefore we need to mathematically well-define the Mellin transform of distributions in the dual space $\mathcal{S}'(\mathbb{R}^+)$. In this paper, we investigate this problem by characterizing the Mellin transform of the Schwartz space $\mathcal{S}(\mathbb{R}^+)$. This allows us to rigorously define the Mellin transform of tempered distributions through a Parseval-type relation. Massless celestial amplitudes are then properly defined by taking the Mellin transform of elements in the topological dual of the Schwartz space $\mathcal{S}(\mathbb{R}^+)$. We conclude the paper with applications to tree-level graviton celestial amplitudes.

\end{abstract}

\end{titlepage}

\setcounter{tocdepth}{2}
\tableofcontents

\newpage
%%%%%%%%%%%%%%%%%%%%%%%%%%%%%%%%%%%%%%%%%%%%%%%%%%%%%%%%%%%%%%%%%
\section{Introduction}
%%%%%%%%%%%%%%%%%%%%%%%%%%%%%%%%%%%%%%%%%%%%%%%%%%%%%%%%%%%%%%%%%
The Mellin transform plays an important role in many fields in mathematics such as the asymptotic study of Gamma related functions. In number theory, it provides a strong tool for the estimation of the coefficients of Dirichlet series. It has also other applications in other areas in mathematics such as the estimation of the asymptotics of integral forms and this list of applications is not exhaustive. However, the importance of the Mellin transform is not only restricted to the field of mathematics, but it has also applications in physics and electrical engineering. We cite as an example the application of the Mellin transform in solving a potential problem in a wedge-shaped geometry where the solution satisfies the Laplace equation with boundary conditions on the edges. In the present paper, we focus on the application of the Mellin transform in celestial holography.

Scattering amplitudes are usually considered in momentum space, with the scattering states defined by the momenta of the particles being scattered. For translation invariant quantum field theories, the amplitudes in momentum space have a simple distributional aspect which consists of the product of a part which is smooth with respect to the external momenta and a momentum-conserving Dirac distribution $\delta^{(d)}(\sum_i p_i)$. In order to compute physical observables such as cross-sections and decay rates, scattering amplitudes are smeared against Gaussian functions or \textit{wave packets} \cite{Peskin:1995ev}. From a mathematical point of view, these are test functions on which the scattering amplitudes act. They belong to the \textit{Schwartz space} $\mathcal{S}(\mathbb{R})$, the space of test functions of rapid decrease. The scattering amplitudes are then defined by duality with the space of Schwartz functions, which in the mathematical literature is known as the space of \textit{tempered distributions} \cite{0f5719a4-0389-38e1-93d8-123291531606}. 

The group of isometries of Minkowski spacetime is the Poincaré group which also contains as a subgroup the Lorentz group. Making this symmetry manifest reveals the conformal properties of scattering amplitudes~\cite{Pasterski:2016qvg} due to the fact that the Lorentz group in $d+2$ dimensions is isomorphic to the conformal group in $d$ dimensions. This indeed plays a pivotal role in the celestial holography program which puts forward a duality between a theory of quantum gravity in the bulk of an asymptotically flat spacetime and a conformal field theory living on the codimension-two celestial sphere. For massless particles the change of basis is performed using the Mellin transform on the energy of that particle. This essentially trades the energy parameter $\omega$ for the weight under boosts and the conformal dimension $\Delta$. It was shown in~\cite{Pasterski:2017kqt} that when the conformal dimension is on the principal continuous series of the Lorentz group $\Delta\in \frac{d}{2}+i\mathbb R$, the states form a basis that is orthogonal and normalizable. Beyond this, the discovery of the relation between soft theorems and asymptotic symmetries~\cite{He:2014laa,Kapec:2014opa,Lysov:2014csa,Campiglia:2014yka,He:2014cra,Kapec:2015vwa,Campiglia:2015yka,Campiglia:2015qka,Kapec:2015ena,Campiglia:2015kxa,Campiglia:2016jdj,Campiglia:2016hvg} (or more generally the IR triangle~\cite{Strominger:2017zoo}\footnote{See review and references therein.}) highlighted the importance of this change of basis in the attempt for flat space holography.

The conformal properties of the scattering amplitude of massless particles become explicit when we perform the Mellin transform on the energies of each of the external particles. We then obtain the so-called \textit{celestial amplitude} which exhibits the properties of a conformal correlation function. These have been studied extensively~\cite{Pasterski:2017ylz,Schreiber:2017jsr,Puhm:2019zbl,Albayrak:2020saa,Arkani-Hamed:2020gyp,Gonzalez:2020tpi,Stieberger:2018edy,Donnay:2023kvm,Banerjee:2017jeg} and in most cases, the computation yields divergent integrals of the form
\be
I_n=\int_0^\infty d\omega~\omega^{n+i\Lambda-1}\, .
\ee
For certain scattering amplitudes, as the one considered in~\cite{Pasterski:2017ylz}, the integrals appearing are of the type $I_0$, which are divergent but recognizable as the Dirac $\delta$-distribution. In more general cases this issue was addressed by analytically continuing the conformal dimensions of the principal series. 
In this paper we rigorously study this change of basis by considering the Mellin transform of tempered distributions. This shows that all massless celestial amplitudes are well-defined by having a finite action on a certain class of functions which we identify as the Mellin transform of the Schwartz space, which is a class of meromorphic functions on the complex plane with simple poles at non-positive integers.

From a mathematical point of view, the extension of the Mellin transform was studied extensively in \cite{Fung}, \cite{Zemanian} and \cite{alma}. In \cite{alma}, the authors proceed by defining the Mellin transform of distributions over the half-line through 
\begin{equation}\label{conventional}
    \mathcal{M}(T)(s)=\langle T,x^{s-1}\rangle.
\end{equation}
In order to give a sense to \eqref{conventional}, a suitable set $\mathcal{I}(a,b)$ of test functions is introduced in \cite{alma} and it consists of all infinitely differentiable functions $\phi$ on $\mathbb{R}^+$ with derivatives rapidly decreasing at both $0$ and $\infty$. The space of test functions $\mathcal{C}^{\infty}_c(\mathbb{R})$ is a subspace of $\mathcal{I}(a,b)$ which implies that $\mathcal{I}'(a,b)$ is a subspace of distributions in $\mathcal{D}'(\mathbb{R}^+)$. Then, \eqref{conventional} allows to define the Mellin transform of distributions $\mathcal{I}'(a,b)$. The resulting function $\mathcal{M}(T)(s)$ for these distributions is holomorphic in the strip $a<\mathcal{R}(s)<b$. The drawback of this method is that it does not include the Mellin transform of distributions for which $\mathcal{M}(T)(s)$ is not necessarily a holomorphic function but rather a distribution as it is the case for some distributions in $\mathcal{S}'(\mathbb{R}^+)$.

The techniques used in \cite{Fung} are based on Fourier analysis from which follow the properties of the Mellin transform of functions in $\mathcal{C}^{\infty}_o(\mathbb{R}^+)$ and distributions in $\mathcal{D}'(\mathbb{R}^+)$. Unfortunately, these methods cannot be used in the study of the Mellin transform of the Schwartz space 
\begin{equation}
    \mathcal{S}(\mathbb{R}^+):=\left\{\phi\in\mathcal{C}^{\infty}(\mathbb{R}^+)\left|\right.~\forall (\alpha,\beta)\in\mathbb{N}^2: ~\sup_{x\in\mathbb{R}^+}\left|x^{\alpha}\phi^{(\beta)}(x)\right|<+\infty\right\}
\end{equation}
and of its dual space. More precisely, the Mellin transform in \cite{Fung} is expressed using the Fourier transform as follows 
\begin{equation}
    \mathcal{M}(\phi)\left(s+it\right)=\mathcal{F}\left(\phi(e^u)e^{u}\right)(t),
\end{equation}
which is well-defined for $\phi$ in $\mathcal{C}^{\infty}_o(\mathbb{R}^+)$, but is ill-defined for functions in $\mathcal{S}(\mathbb{R}^+)$. 

Hence, we proceed in this paper by considering the Mellin transform independently from an a priori knowledge of Fourier or Laplace transforms and we provide a general method that allows to define the Mellin transform of any distribution in $\mathcal{S}'(\mathbb{R}^+)$ through 
\begin{equation}
\langle\mathcal{M}(T),{f}\rangle_{{{\tilde{\mathcal M}}^{+\prime}},{\tilde{\mathcal M}}^+}=\langle T,\mathcal{M}^{-1}(f)\rangle_{\mathcal S'(\mathbb R^+),\mathcal S(\mathbb R^+)}\, ,
\end{equation}
where ${\tilde{\mathcal M}}^+$ is the Mellin transform of the Schwartz space $\mathcal{S}(\mathbb{R}^+)$ and $\tilde{\mathcal M}^{+\prime}$ its topological dual.

The paper is organized as follows: we start by introducing the notations and conventions which are used throughout the paper. In section~\ref{sec2}, we characterize the Mellin transform of the Schwartz space $\mathcal S(\mathbb R^+)$ which we denote by $\tilde{\mathcal M}^+$ and establish that the Mellin transform defines an isomorphism between the spaces $\mathcal S(\mathbb R^+)$ and $\tilde{\mathcal M}^+$ by a precise correspondence between individual terms in the asymptotic expansion of an original function and singularities of its Mellin transform.  In section~\ref{sec3}, we use Parseval's relation for the Mellin transform to define the dual space of $\tilde{\mathcal M}^{+}$. Massless celestial amplitudes belong to $\tilde{\mathcal M}^{+\prime}$ and have a well-defined bracket with functions in $\tilde{\mathcal M}^+$. We conclude the paper with section~\ref{sec4} in which we rigorously compute the tree-level graviton celestial amplitudes by applying the methods developed in section~\ref{sec3}. We finally end the paper with some remarks and future perspectives.

%%%%%%%%%%%%%%%%%%%%%%%%%%%%%%%%%%
\subsection*{Notations and Conventions}
In the sequel, we use the following notations 
\begin{equation}
    \mathbb{N}^{0}:=\mathbb{N}\setminus \{0\},~~\mathbb{R}^+:=\left[0,+\infty\right),~~\mathbb{R}^{+\star}:=\mathbb{R}^+\setminus \{0\},~~\mathbb Z^-:=\mathbb Z\setminus \mathbb N^0.
\end{equation}
We also use the notation
\be
\left\|\phi\right\|_\infty=\sup_{x\in\mathbb R^+}|\phi(x)|\, .
\ee
The following Fourier and inverse Fourier transforms conventions are also used: given $\phi$ in $L^2(\mathbb{R})$, the Fourier transform of $\phi$ is defined as follows 
\begin{equation}
    \mathcal{F}(\phi)(\xi):=\int_{\mathbb{R}}\phi(x)e^{-ix\xi}~dx.
\end{equation}
Similarly, its inverse Fourier transform is defined as
\begin{equation}
    \mathcal{F}^{-1}(\phi)(x):=\frac{1}{{2\pi}}\int_{\mathbb{R}}\phi(\xi)e^{ix\xi}~d\xi.
\end{equation}
Given a function $\phi$ that satisfies
$$\int_{\mathbb R^+}dx~|\phi(x)|x^{c-1}<+\infty,$$ we define its Mellin transform as follows
\begin{equation}
    \mathcal{M}(\phi)(c+it)=\int_{\mathbb{R}^+}\phi(x)~x^{c+it-1}~dx~.
\end{equation}
The definition domain of a Mellin transform is a strip. Hence, we introduce the notation $$St( \alpha,\beta):=\left\{s\in\mathbb{C},~\alpha<\mathcal{R}(s)<\beta\right\},$$
where $\mathcal{R}(s)$ is the real part of $s$. We introduce the Schwartz spaces
\begin{equation}
    \mathcal{S}(\mathbb{R}):=\left\{\phi\in \mathcal{C}^{\infty}(\mathbb{R})~|\forall(\alpha,\beta)\in\mathbb{N}^2:~\sup_{x\in\mathbb{R}}\left|x^{\alpha}\partial^{\beta}_x\phi(x)\right|<+\infty\right\}
\end{equation} 
and
\begin{equation}
    \mathcal{S}\left(\mathbb{R}^{+}\right):=\left\{\chi^+(x)\phi(x)\left|\right. \phi \in\mathcal{S}(\mathbb{R})\right\},
\end{equation}
where $\chi^+$ denotes the characteristic function of the semi-line $\mathbb{R}^{+}$. We denote the dual of $\mathcal{S}(\mathbb{R}^+)$ by $\mathcal{S}'(\mathbb{R}^+)$, and call it the Schwartz space of tempered distributions over the half-line. For an extensive topological study of the space $\mathcal{S}(\mathbb{R}^+)$ and its dual space $\mathcal{S}'(\mathbb{R}^+)$, we refer the reader to \cite{Pott}.
%%%%%%%%%%%%%%%%%%%%%%%%%%%%%%%%%%%%%%%%%%%%%%%%%%%%%%%%%%%%%%%%%

\section{The Mellin Transform of the Schwartz Space}\label{sec2}
In order to define massless celestial amplitudes, we need to define the space of test functions against which they will be smeared. In this section, we identify the space of conformal weight test functions as the Mellin transform of the Schwartz space $\mathcal{S}(\mathbb{R}^+)$. The main result of this section is theorem~\ref{mainres}. First, we establish some basic properties of the Mellin transform $\mathcal{M}(\phi)$ with $\phi$ in $\mathcal{S}(\mathbb{R}^+)$, that we gather in the following proposition:
\begin{proposition}\label{prop1}
   Given $\phi$ in $\mathcal{S}(\mathbb{R}^+)$, the following properties hold:
    \begin{itemize}
        \item[i)] The fundamental strip\footnote{defined as the largest open strip on which the Mellin transform $\mathcal M(\phi)$ is defined.} of $\mathcal{M}(\phi)$ is $St(0,+\infty)$ and $\mathcal{M}(\phi)$ is holomorphic in $St( 0,+\infty)$.
        \item[ii)] For all $N\in\mathbb{N}$, $\mathcal{M}(\phi)(s)$ has a meromorphic continuation to $\mathcal{R}(s)>-N$ with simple poles at $s=-n$ with $0\leq n \leq N-1$ and no other singularities~\cite{ZagierAppendixTM}.
        \item[iii)] For $c>0$ we have
        \begin{equation}
            \mathcal{M}(\phi)\left(c+i\cdot\right)\in\mathcal{S}(\mathbb{R})
        \end{equation}
        and for $c\in\mathbb{R}^-\setminus \mathbb{Z}^-$ 
        \begin{equation}
            t^{\alpha}\mathcal{M}(\phi)\left(c+it\right)
        \end{equation}
    is uniformly bounded with respect to $t\in\mathbb{R}$ for all $\alpha\in\mathbb{N}$.
    \end{itemize}
\end{proposition}
\begin{proof}
\begin{itemize}
    \item First, we prove i). Given $\phi\in\mathcal{S}(\mathbb{R}^+)$, we have for $s=c+it$
    \begin{equation}
        \left|\int_{\mathbb{R}^+}\phi(x)~x^{s-1}dx\right|\leq \int_{\mathbb{R}^+}|\phi(x)|~x^{c-1}dx.
    \end{equation}
    Performing the change of variable $u=\log x$, we write 
    \begin{equation}\label{3.7}
         \int_{\mathbb{R}^+}|\phi(x)|~x^{c-1}dx=\int_{\mathbb{R}}~|\phi(e^u)|e^{cu}du.
    \end{equation}
    Decomposing the integral in the right-hand side of \eqref{3.7} over the regions $\mathbb{R}^+$ and $\mathbb{R}^-$, we obtain for $c>0$ 
    \begin{equation}\label{1.9}
       \left|\mathcal{M}(\phi)(c+it)\right| \leq\frac{1}{c}\sup_{x\in\mathbb{R}^+}|\phi(x)|+\sup_{x\in\mathbb{R}^+}|\phi(x)x^{[c]+2}|<+\infty.
    \end{equation}
     This proves that $\mathcal{M}(\phi)$ is well-defined for $c>0$. In order to establish that $\mathcal{M}(\phi)$ is holomorphic on $St(0,+\infty)$, we consider a triangle $\vartriangle\subset St(0,+\infty)$ parameterized by $\omega:[0,1]\rightarrow \vartriangle$ with $\omega$ being piecewise continuously differentiable. The function $(u,y)\longmapsto \phi(u)u^{\omega(y)-1}\omega'(y),~u\in\mathbb{R}^+,~y\in[0,1]$ is measurable with respect to the
two-dimensional Lebesgue measure, and furthermore we have using \eqref{1.9}
\begin{multline}\label{1.10}
\int_0^1\int_{\mathbb{R}^+}\left|\phi(x)\right|~x^{\omega(y)-1}~\left|\omega'(y)\right|dx~dy\\\leq \left(\int_0^1 \left|\frac{\omega'(y)}{\omega(y)}\right|dy~\left\|\phi\right\|_{\infty}+\sup_{x\in\mathbb{R}^+}\left|\phi(x)x^{[\omega(\tilde{y})]+2}\right|~\int_0^1|\omega'(y)|dy\right)
\end{multline}
with $\omega(\tilde{y}):=\sup_{y\in[0,1]}\omega (y)$. Remembering that $\omega(y)\in \vartriangle\subset St(0,+\infty)$, we then deduce 
\begin{equation}
    \int_0^1\int_{\mathbb{R}^+}\left|\phi(x)\right|~x^{\omega(y)-1}~\left|\omega'(y)\right|dx~dy<+\infty.
\end{equation}
Hence the assumptions of Fubini's theorem are satisfied and we write
\begin{align*}
    \int_{\vartriangle} \mathcal{M}(\phi)(s)~ds&=\int_0^1 \mathcal{M}(\phi)(\omega(y))\omega'(y)~dy\\
    &=\int_{\mathbb{R}^+}\phi(x)\left(\int_0^1~x^{\omega(y)-1}\omega'(y)dy\right)dx\\
    &=\int_{\mathbb{R}^+}\phi(x)\left(\int_{\vartriangle}x^{s-1}ds\right)dx=0,
\end{align*}
where we used Cauchy's theorem together with the fact that $u^{s-1}$ is holomorphic in $St(0,+\infty)$ for all $u>0$. Using lemma 2 from~\cite{Butzer1997}, we also have that $\mathcal M(\phi)$ is continuous on its fundamental  strip $St(0,\infty)$. Since $\vartriangle\subset St(0,+\infty)$ is arbitrary, we deduce by Morera's theorem that $\mathcal{M}(\phi)$ is holomorphic on $St(0,+\infty)$.
\item The proof of ii) follows the same line of reasoning used in~\cite{ZagierAppendixTM} in which $\mathcal M (\phi)$ was shown to admit a meromorphic extension on $\mathbb C$. Given $N\in\mathbb N$, we perform a Taylor expansion of $\phi$ in $\mathcal M (\phi)$ as follows
\begin{equation}\label{Taylor2.8}
\mathcal M(\phi)(c+it) = \int_0^1 \left( \phi(x) -\sum_{n=0}^{N-1}a_n x^n\right)x^{c+it-1}dx+\sum_{n=0}^{N-1}\frac{a_n}{n+c+it} +\int_1^{+\infty} \phi(x) x^{c+it-1}dx\,,
\end{equation}
where $a_n:=\frac{\phi^{(n)}(0)}{n!}$ and $\phi^{(n)}$ is the $n^{th}$-order derivative of $\phi$. Since $\phi$ is in $\mathcal S(\mathbb{R^+})$, we have by Taylor-Lagrange's inequality
$$
 \left|\phi(x) -\sum_{n=0}^{N-1}a_n x^n\right|\leq \|\phi^{(n)}\|_{\infty} \ x^N,
$$
which implies that the first integral is well-defined for all $c>-N$. The second term contains the poles $c+it=-n$ for $n\in{0,\cdots,N-1}$ and the third term is well-defined for all $ (c,t) \in\mathbb R^2$. This holds for all $N\in\mathbb N$ which directly gives that $\mathcal M(\phi)$ has a meromorphic extension to $\mathbb C$ with simple poles at $s=c+it\in\mathbb Z^-$.
\item In this part, we establish iii). Given $c>0$, we prove first that 
$$\mathcal{M}(\phi)\left(c+i\cdot\right)\in\mathcal{S}(\mathbb{R}).$$
Lemma 2 in~\cite{Butzer1997} gives that $\mathcal{M}(\phi)$ is a continuous function on the line $\{c\}\times i\mathbb R$. Now, let us verify that for $c>0$, the Mellin transform $\mathcal{M}(\phi)(c+it)$ is infinitely differentiable with respect to $t$ and furthermore we have for all $\beta\in\mathbb N$
\begin{equation}\label{1200}
    \partial^{\beta}_t\mathcal{M}(\phi)\left(c+it\right)=i^{\beta}\int_{\mathbb{R}^+}\phi(x)(\log x)^{\beta}x^{c+it-1}dx.
\end{equation}
Let $\mathcal{C}_{\epsilon}(it)$ be the circle with radius $\epsilon<1$ and origin $s=it$. Using  Cauchy's integral formula for derivatives we have for $u>0$
\begin{align}\label{121}
   (\log u)^{\beta}~u^{it}&=(-i)^{\beta}\partial^{\beta}_tu^{it}=\frac{\beta!}{2\pi i}\int_{\mathcal{C}_{\epsilon}(it)}\frac{u^{w}}{(w-it)^{\beta+1}}dw\nonumber\\
   &=\frac{\beta!}{2\pi i}\int_0^{2\pi}\frac{u^{\psi(y)}}{\left(\psi(y)-it\right)^{\beta+1}}\psi'(y)dy,
\end{align}
 where $\psi(y)=it+\epsilon e^{iy}$. We write 
 \begin{equation}\label{hope}
\int_{\mathbb{R}^+} 
     \left|\phi(u)\right|~u^{c-1+\epsilon\cos y}du=\int_{0}^1 
     \left|\phi(u)\right|~u^{c-1+\epsilon\cos y}du+\int_{1}^{+\infty}  
     \left| u^{c-1+\epsilon\cos y}~\phi(u)\right|du.
\end{equation}
\eqref{hope} is bounded by 
\begin{equation*}
    \|\phi\|_{\infty}  
   \int_{0}^1 
     u^{c-1+\epsilon \cos y}du+\int_{1}^{+\infty} 
     u^c\left|\phi(u)\right|du=\frac{ \|\phi\|_{\infty}}{c+\epsilon\cos y}+\int_{1}^{+\infty} 
     u^c\left|\phi(u)\right|du.
\end{equation*}
Hence, we deduce 
\begin{equation}
\int_{\mathbb{R}^+}\left|\phi(u)~u^{c-1+\psi(y)}\right|du\leq K
\end{equation}
with $K:=\frac{ \|\phi\|_{\infty}}{c-\epsilon}+\int_{\mathbb{R}^+}  u^c\left|\phi(u)\right|du$. Thus we obtain for $\beta\in\mathbb{N}$
 \begin{align*}
\int_0^{2\pi}\int_{\mathbb{R}^+}\left|\phi(u)\frac{u^{\psi(y)}}{\left(\psi(y)-it\right)^{\beta+1}}\psi'(y)\right| du dy 
     &\leq K~\int^{2\pi}_0 \left|\frac{\psi'(y)}{\left(\psi(y)-it\right)^{\beta+1}}\right| dy
     \\& \leq 2\pi K~ \epsilon^{-\beta}.
     \end{align*}
Therefore, the hypotheses of Fubini's theorem are fulfilled and we have by \eqref{121}
     \begin{equation}\label{hope2}
         \int_{\mathbb{R}^+}\phi(x)(\log x)^{\beta}x^{c+it-1}dx=\frac{\beta!}{2\pi i}\int_0^{2\pi}\frac{\psi'(y)}{\left(\psi(y)-it\right)^{\beta+1}}\left(\int_{\mathbb{R}^+}\phi(u)u^{c+\psi(y)-1}du\right)dy.
     \end{equation}
For $c>0$, we have 
\begin{equation}
   \mathcal{M}(\phi)\left(c+\psi(y)\right)= \int_{\mathbb{R}^+}\phi(u)~u^{c+\psi(y)-1}du.
\end{equation}
Hence, the right-hand side of \eqref{hope2} is given by 
\begin{equation}
    \frac{\beta!}{2\pi i}\int_0^{2\pi}\frac{\psi'(y)}{\left(\psi(y)-it\right)^{\beta+1}}~\mathcal{M}(\phi)\left(c+\psi(y)\right)dy.
\end{equation}
Applying again Cauchy's integral formula leads directly to \eqref{1200} for any $\beta\in\mathbb{N}$. Now, let us verify for $c>0$
$$\sup_{t\in\mathbb R}\left|t^{\alpha}\partial^{\beta}_t\mathcal{M}(\phi)\left(c+it\right)\right|<+\infty.$$
 Performing the change of variable $x=e^u$ in \eqref{1200}, we have
\begin{equation}\label{2166}
t^{\alpha}\partial_t^{\beta}\mathcal{M}\left(\phi\right)\left(c+it\right)=i^{\beta-\alpha}\int_{\mathbb{R}}\phi(e^u)u^{\beta}e^{cu} \partial^{\alpha}_u(e^{itu})du.
\end{equation}
Using the relation
\begin{equation}\label{hope217}
\partial^{k}_u\left(\phi(e^u)\right)=\sum_{p=1}^k \alpha_k(p)~e^{pu}\phi^{(p)}(e^u)
\end{equation}
with $\alpha_k(p)$ denoting a positive constant, together with integration by parts and the general Leibniz formula, we deduce 
\begin{align}
 \int_{\mathbb{R}}\phi(e^u)&u^{\beta}e^{cu} \partial^{\alpha}_u(e^{itu})du
    =(-1)^{\alpha}\int_{\mathbb{R}}\partial^{\alpha}_u\left(\phi(e^u)u^{\beta}e^{cu}\right) e^{itu}du\nonumber\\
 &= (-1)^{\alpha}\sum_{k=0}^{\alpha}\sum_{n=0}^{\alpha-k}\sum_{p=1-\delta_{k,0}}^{k}\gamma(\alpha,k,p,n,\beta)\chi_{\{\beta\geq n\}}~ c^{\alpha-k-n}\int_{\mathbb{R}}e^{(c+p)u}e^{itu}\phi^{(p)}(e^u)u^{\beta-n}du,\label{2177}
\end{align}
where 
$$\chi_{\{\beta\geq n\}}= \left\{
    \begin{array}{ll}
        1 & \mbox{if } \beta \geq n \\
        0 & \mbox{otherwise}
    \end{array}
\right.$$
and 
$$\gamma(\alpha,k,p,n,\beta):=\binom{\alpha}{k}\binom{\alpha-k}{n}\frac{\beta!}{(\beta-n)!}\alpha_k(p)~.$$
We use the bound 
$$u\leq e^u,~~~\forall u \in \mathbb{R}^+$$
to obtain for $p\leq \alpha$
\begin{equation}
\left|\int_{\mathbb{R}^+}e^{(c+p)u}e^{itu}\phi^{(p)}(e^u)u^{\beta-n}du\right|\leq \frac{1}{c}\sup_{{x\in\mathbb{R}^+}}\left|x^{(2c+\alpha+\beta)}\phi^{(p)}(x)\right|.
\end{equation}
We also have  
\begin{align}\label{2200}
\left|\int_{\mathbb{R}^-}e^{(c+p)u}e^{itu}\phi^{(p)}(e^u)u^{\beta-n}du\right|&\leq \frac{2}{c}\sup_{u\in\mathbb{R}^+}\left|e^{-\frac{cu}{2}}u^{\beta-n}\right|\left\|\phi^{(p)}\right\|_{\infty}\nonumber\\&\leq 2\left(c^{-1}+\sup_{u\in\mathbb{R}^+}\left|e^{-\frac{u}{2}}u^{\beta}\right|c^{n-\beta-1}\right)~\left\|\phi^{(p)}\right\|_{\infty}\nonumber\\
&\leq O(1)\left(c^{-1}+c^{n-\beta-1}\right)~\left\|\phi^{(p)}\right\|_{\infty},
\end{align}
where $O(1)$ is a constant that depends only on $\beta$. Combining \eqref{2166}-\eqref{2200}, we deduce
\begin{multline}
    \left|t^{\alpha}\partial^{\beta}_t\mathcal{M}(\phi)\left(c+it\right)\right|\leq O(1)~\left\|\gamma\right\|_{\infty} \left(\sum_{k=0}^{\alpha}c^{k-1}\right)\\\times\left\{\sup_{{x\in\mathbb{R}^+}}\left|x^{(2c+\alpha+\beta)}\phi^{(p)}(x)\right|+(1+c^{-\beta})\left\|\phi^{(p)}\right\|_{\infty}\right\}<+\infty,
\end{multline}
where $\left\|\gamma\right\|_{\infty}:=\sup_{0\leq k,p,n\leq \alpha}\left|\gamma(\alpha,k,p,n,\beta)\right|$ and $O(1)$ is a positive constant that depends only on $\alpha$ and $\beta$. 
\item In this part of the proof, we verify the last point of proposition \ref{prop1} that is for $c\in\mathbb{R}^-\setminus \mathbb{Z}^-$ 
        \begin{equation}
            t^{\alpha}\mathcal{M}(\phi)\left(c+it\right)
        \end{equation}
    is uniformly bounded with respect to $t\in\mathbb{R}$ for all $\alpha\in\mathbb{N}$.\\
For $\alpha\in\mathbb{N}^*$ and $c>0$, we have 
\begin{align}\label{1.13}
     t^{\alpha}~\mathcal{M}(\phi)(c+it)
     &=(-i)^{\alpha}\int_{\mathbb{R}}\phi(e^u)~e^{cu}~\partial^{\alpha}_u(e^{itu})du\nonumber\\
     &=i^{\alpha}\int_{\mathbb{R}}\partial^{\alpha}_u\left(\phi(e^u)~e^{cu}\right)e^{itu}~du~.
\end{align}
Using \eqref{hope217} together with the general Leibniz formula, we obtain that \eqref{1.13} is equal to 
\begin{equation}
    i^{\alpha}\sum_{k=0}^{\alpha}\sum_{p=1-\delta_{k,0}}^k \beta(\alpha,k,p)~c^{\alpha-k}\int_{\mathbb{R}}\phi^{(p)}(e^u)e^{(c+p)u}e^{itu}~du~,
\end{equation}
where $\beta(\alpha,k,p):=\binom{\alpha}{k}\alpha_k(p)$. Hence, we deduce for all $c>0$
\begin{equation}
    t^{\alpha}~\mathcal{M}(\phi)(c+it)=i^{\alpha}\sum_{k=0}^{\alpha}\sum_{p=1-\delta_{k,0}}^k \beta\left(\alpha,k,p\right)~c^{\alpha-k}\mathcal{M}(\phi^{(p)})(c+p+it).
\end{equation}
By the uniqueness of the meromorphic continuation of $\mathcal{M}(\phi)$, we deduce that 
\begin{equation}\label{2.26}
   t^{\alpha}~\mathcal{M}(\phi)(c+it)= i^{\alpha}\sum_{k=0}^{\alpha}\sum_{p=1-\delta_{k,0}}^k \beta\left(\alpha,k,p\right)c^{\alpha-k}\mathcal{M}(\phi^{(p)})(c+p+it)
\end{equation}
holds for all $c\in \mathbb{R}\setminus\mathbb{Z}^-$. If $c+p>0$, we use the bound \eqref{1.9} to obtain
\begin{equation}
    \left|\mathcal{M}(\phi^{(p)})(c+p+it)\right|\leq \frac{1}{c+p}\left\|\phi^{(p)}\right\|_{\infty}+\sup_{x\in\mathbb{R}^+}\left|\phi^{(p)}(x)x^{[c+p]+2}\right|~,
\end{equation}
which can be bounded remembering that $p\leq \alpha$ by 
\begin{equation}\label{1.14}
    2\max\left(\frac{1}{c+p},1\right)\sup_{1\leq p\leq \alpha }\left\|\phi^{(p)}\right\|_{\infty}+\sup_{\substack{x\in\mathbb{R}^+\\ 1\leq p\leq \alpha}}\left|\phi^{(p)}(x)x^{[c+\alpha]+2}\right|.
\end{equation}
Otherwise, we perform a Taylor expansion of $\phi^{(p)}$ around $0$ and we write 
\begin{multline}
   \mathcal{M}(\phi^{(p)})(c+p+it)=\int_0^1\left(\phi^{(p)}(x)-\sum_{n=0}^{N-1}\frac{\phi^{(p+n)}(0)}{n!} x^n\right)x^{c+p+it-1}dx+\sum_{n=0}^{N-1}\frac{\phi^{(p+n)}(0)}{n!(c+p+it+n)}\\+\int_1^{+\infty}\phi^{(p)}(x)~x^{c+p+it-1}dx,
\end{multline}
where the positive integer $N$ verifies $N+c+p>0$. Using  Taylor-Lagrange's inequality and remembering that $\phi\in\mathcal{S}(\mathbb{R}^+)$, we deduce
\begin{equation}
    \left|\mathcal{M}(\phi^{(p)})(c+p+it)\right|\leq \frac{\left\|\phi^{(p+N)}\right\|_{\infty}}{N+c+p}+\sum_{n=0}^{N-1}\frac{|\phi^{(p+n)}(0)|}{\sqrt{(c+p+n)^2+t^2}}+\frac{\sup_{\substack{x\in\mathbb{R}^+}}\left|x^m\phi^{(p)}(x)\right|}{\left|c+p-m\right|},~~\forall m\in\mathbb{N}.
\end{equation}
Hence, we obtain for $N=-[c+p]+2$ and $m=0$ the following bound
\begin{equation}\label{2.29}
\begin{aligned}
    \left|\mathcal{M}(\phi^{(p)})(c+p+it)\right|\leq &O(1)~\sup_{0\leq n\leq \alpha+[-c+2]}\left\|\phi^{(n)}\right\|_{\infty}\\
    &\left(-[c+p]+2\right)\max\left(\frac{1}{|c+p-[c+p]+1|},\frac{1}{|c+p-[c+p]|}\right).
    \end{aligned}
\end{equation}
Using \eqref{1.14} for the particular case of $c>0$, we deduce
\begin{multline}
    \left|t^{\alpha}\mathcal{M}(\phi)(c+it)\right|\leq O(1)~\left(\sum_{k=0}^{\alpha} c^{\alpha-k}\right)\\\left\{ \max\left(\frac{1}{c},1\right)\sup_{1\leq p\leq \alpha }\left\|\phi^{(p)}\right\|_{\infty}+\sup_{\substack{x\in\mathbb{R}^+\\ 1\leq p\leq \alpha}}\left|\phi^{(p)}(x)x^{[c+\alpha]+2}\right|\right\},
\end{multline}
where $O(1)$ denotes a positive constant that depends only on $\alpha$.
For $c\in\mathbb{R}^-\setminus\mathbb{Z}^-$, we combine \eqref{2.26}, \eqref{1.14} and \eqref{2.29} to deduce 
\begin{multline}
    \left|t^{\alpha}\mathcal{M}(\phi)\left(c+it\right)\right|\leq O(1)~\left(\sum_{k=0}^{\alpha} c^{\alpha-k}\right)\\ \left\{\sum_{p=0}^{\alpha}\left(-[c+p]+2\right)\max\left(\frac{1}{|c+p-[c+p]+1|},\frac{1}{|c+p-[c+p]|}\right)\sup_{0\leq n\leq \alpha+[-c+2]}\left\|\phi^{(n)}\right\|_{\infty}\right.\\\left. +\sup_{\substack{x\in\mathbb{R}^+\\ 1\leq n\leq \alpha}}\left|\phi^{(n)}(x)x^{[c+\alpha]+2}\right|\right\}.
\end{multline}
\end{itemize}
\end{proof}
\noindent In the proof of Theorem \ref{Thm1}, we need the fundamental Lemma of Riemann-Lebesgue in the frame of Mellin transforms, which in our case we consider for particular values of $\mathrm{Re}(s)\in\mathbb{Z}^-$:
\begin{lemma}
   If $\phi\in\mathcal{S}(\mathbb{R}^+)$, then 
   \begin{equation}
       \lim_{|t|\rightarrow +\infty}\mathcal{M}\left(\phi\right)\left(c+it\right)\longrightarrow 0,~~~~~\forall c\in\mathbb{R}^-\setminus\{0\}~.
   \end{equation}
   For $c=0$, we have 
   \begin{equation}\label{spe235}
       \lim_{t\rightarrow +\infty}t~\mathcal{M}\left(\phi\right)\left(it\right)\longrightarrow 0.
       \end{equation}
\end{lemma}
\begin{proof}
\begin{itemize}
   \item We proceed as in \cite{Butzer1997} by using the same technique that proves the Lebesgue-Riemann Lemma by introducing the Mellin translation operator $\tau^c_{h}$ for $\phi:\mathbb{R}^+\rightarrow \mathbb{C}$, $c\in\mathbb{R}$ and $h\in\mathbb{R}^+$ defined by
    \begin{equation}
        \left(\tau_h^c\phi\right)(x):=h^c\phi(hx),~~~~x\in\mathbb{R}^+.
    \end{equation}
    For $s=c+it$, $t\in\mathbb{R}^*$ and $h:=e^{-\frac{\pi}{t}}$, we have 
    \begin{equation}
        -\mathcal{M}(\phi)(s)=e^{i\pi}\mathcal{M}(\phi)(s)=h^{-it}\mathcal{M}(\phi)(s)=\mathcal{M}(\tau_h^c\phi)(s).
    \end{equation}
    Using the linearity of the Mellin transform, we obtain 
    \begin{equation}\label{spe237}
        2~\mathcal{M}(\phi)(s)=\mathcal{M}(\phi)(s)-\mathcal{M}(\tau_h^c\phi)(s)=\mathcal{M}(\phi-\tau_h^c\phi)(s).
    \end{equation}
    We write 
    $$\Phi_{h,c}:=\phi-\tau_h^c\phi.$$
    Performing a Taylor expansion of $\Phi_{h,c}$ around $0$, we obtain for $c<0$
    \begin{multline}
   \mathcal{M}(\Phi_{h,c})(c+it)=\int_0^1\left(\int_0^1\frac{(1-t)^{N-1}}{(N-1)!}\partial^N_t\Phi_{h,c}(tx)dt\right)x^{c+it-1}dx+\sum_{n=0}^{N-1}\frac{\Phi_{h,c}^{(n)}(0)}{n!(c+it+n)}\\+\int_1^{+\infty}\Phi_{h,c}(x)~x^{c+it-1}dx,
\end{multline}
where $N$ is a positive integer fixed such that $N+c>1$. We have 
\begin{equation}\label{spe1}
    \Phi_{h,c}^{(n)}(0)=\left(1-h^{c+n}\right)\phi^{(n)}(0).
\end{equation}
Since $h^{c+n}=e^{-\frac{\pi(c+n)}{t}}\xrightarrow[|t|\rightarrow \infty]{}1$, we deduce that \eqref{spe1} converges to $0$ if $|t|\rightarrow \infty$.\\
We rewrite  
\begin{equation}\label{spe2}
   \int_1^{+\infty}\Phi_{h,c}(x)~x^{c+it-1}dx~,
\end{equation}
as follows 
\begin{equation}\label{spe241}
    \left(1-h^c\right)\int_1^{+\infty}\phi(x)x^{c+it-1}dx+h^c\left(\int_1^{+\infty}\left(\phi(x)-\phi(hx)\right)x^{c+it-1}dx\right).
\end{equation}
The first term on the right-hand side converges to $0$ if $|t|\rightarrow \infty$. For the second term, we have 
\begin{equation}
   \left|\phi(x)-\phi(hx)\right|x^{c-1}\leq 2\left\|\phi\right\|_{\infty}x^{c-1},
\end{equation}
Since $\phi$ is in $\mathcal{S}(\mathbb{R}^+)$, then $\phi(hx)$ converges pointwise to $\phi(x)$ if $|t|\rightarrow \infty$. By Lebesgue's dominated convergence theorem we obtain 
\begin{equation}
    h^c\left(\int_1^{+\infty}\left(\phi(x)-\phi(hx)\right)x^{c+it-1}dx\right)\xrightarrow[|t|\rightarrow \infty]{} 0.
\end{equation}
Hence, \eqref{spe2} converges to $0$. In order to conclude that 
\begin{equation}\label{spe244}
    \mathcal{M}(\Phi_{h,c})(c+it)\xrightarrow[|t|\rightarrow\infty]{} 0,
\end{equation}
we need to establish that 
\begin{equation}\label{spe245}
    \int_0^1\left(\int_0^1\frac{(1-t)^{N-1}}{(N-1)!}\partial^N_t\Phi_{h,c}(tx)dt\right)x^{c+it-1}dx\xrightarrow[|t|\rightarrow \infty]{}0.
\end{equation}
We have
\begin{equation}\label{spe246}
    \int_0^1{(1-t)^{N-1}}\partial_t^N\Phi_{h,c}(tx)~dt= \int_0^1{(1-t)^{N-1}}x^{N}\left(\phi^{(N)}(tx)-h^{N+c} \phi^{(N)}(thx)\right)dt.
\end{equation}
Proceeding as in \eqref{spe241}, we decompose \eqref{spe246} as follows
\begin{equation}\label{spe247}
\left(1-h^{N+c}\right)x^N \int_0^1{(1-t)^{N-1}}\phi^{(N)}(tx)~dt+h^{N+c} x^N\int_0^1{(1-t)^{N-1}}\left(\phi^{(N)}(tx)-\phi^{(N)}(thx)\right)dt.
\end{equation}
Choosing $N$ such that $N+c-1>1$ and using the bound  
\begin{equation}
    \left|\int_0^1\frac{(1-t)^{N-1}}{(N-1)!}\phi^{(N)}(tx)~dt\right|\leq \frac{\|\phi^{(N)}\|_{\infty}}{N!}~,
\end{equation}
we deduce that 
\begin{equation}
    \left|\int_0^1{\frac{{(1-t)^{N-1}}}{{(N-1)!}}}\left(\phi^{(N)}(tx)-\phi^{(N)}(thx)\right)dt\right|\leq 2\frac{\|\phi^{(N)}\|_{\infty}}{N!}.
\end{equation}
Again by the Lebesgue's dominated convergence theorem, we obtain 
\begin{equation}
    \int_0^1\int_0^1\frac{(1-t)^{N-1}}{(N-1)!}\left(\phi^{(N)}(tx)-\phi^{(N)}(thx)\right)dt~x^{c+it-1{+N}}dx\xrightarrow[|t|\rightarrow\infty]{}0.
\end{equation}
This together with \eqref{spe247} imply \eqref{spe245}. Combining \eqref{spe244} with \eqref{spe237}, we obtain 
\begin{equation}
    \mathcal{M}(\phi)(c+it)\xrightarrow[|t|\rightarrow\infty]{}0,~~~~~\forall c\in\mathbb{R}^-\setminus\{0\}.
\end{equation}
\item For $c=0$ and $t\in\mathbb{R}^*$, the Mellin transform $\mathcal{M}(\phi)(c+it)$ is well-defined through \eqref{Taylor2.8}. Therefore, we write 
\begin{align}
    t~\mathcal{M}(\phi)(it)&=-i\int_{\mathbb{R}}\phi(e^u)\partial_u(e^{itu})du\nonumber\\
    &=i\int_{\mathbb{R}}\phi'(e^{u}) e^u e^{itu}du\nonumber\\
    &=i~\mathcal{M}(\phi')\left(1+it\right).\label{UpropU}
\end{align}
By proposition \ref{prop1}, \eqref{UpropU} converges to $0$ if $|t|\rightarrow +\infty$. This proves \eqref{spe235}.
\end{itemize}
\end{proof}
\subsection{The Inversion Formula for the Mellin Transform}
In this subsection, we prove the inversion formula of the Mellin transform for functions in the Schwartz space $\mathcal{S}(\mathbb{R}^+)$. Before stating the inversion Mellin formula Theorem, we need the following Lemma:
\begin{lemma}\label{lemme1}
    For $f$ in $\mathcal{S}(\mathbb{R}^+)$ and $c>0$, we have 
    \begin{equation}\label{Bp6}
        e^{cx}f(e^x)\in \mathcal{S}(\mathbb{R})~.
    \end{equation}
\end{lemma}
\begin{proof} 
 Let $f$ in $\mathcal{S}(\mathbb{R}^+)$, $c>0$ and $n\in\mathbb{N}$, we bound uniformly 
    $$x^n e^{cx}f(e^x).$$
    If $x\leq 0$, then clearly we have 
    \begin{equation}\label{Bp7}
    \left|x^n e^{cx}f(e^x)\right|\leq \sup_{x\in\mathbb{R}^+}x^n e^{-cx}~\|f\|_{\infty}.
    \end{equation}
    For $x>0$, we have
    \begin{equation}\label{Bp8}
        \left|x^n e^{cx}f(e^x)\right|\leq  e^{(c+1)x}f(e^x) \leq \sup_{x\in\mathbb{R}^+}\left|x^{[c]+2}f(x)\right|~.
    \end{equation}
    Combining \eqref{Bp7} and \eqref{Bp8} gives \eqref{Bp6}.
  \end{proof}
\begin{theorem}\label{Thm1}(Inversion formula)
    Given $\phi$ in $\mathcal{S}(\mathbb{R}^+)$ and $\mathcal{M}(\phi)$ its Mellin transform. We have 
    \begin{equation}
        \phi(x)=\frac{x^{-c}}{2\pi}\int_{\mathbb{R}}\mathcal{M}(\phi)(c+it)~x^{-it}~dt,~~\forall c>0,~~\forall x>0.
    \end{equation}
\end{theorem}
\begin{proof}
    Let $\phi$ in $\mathcal{S}(\mathbb{R}^+)$ and $c>0$. From proposition \ref{prop1}, the Mellin transform
    \begin{equation}\label{B4}
        \int_{\mathbb{R}^+}\phi(u)~u^{s-1}~du
    \end{equation}
     is well-defined. Using (\ref{B4}) we write for $x>0$
     \begin{align*}
         \frac{x^{-c}}{2\pi}\int_{\mathbb{R}}\mathcal{M}(\phi)(c+it)~x^{-it}~dt&=\frac{x^{-c}}{2\pi }\int_{\mathbb{R}}\int_{\mathbb{R}^+}\phi(u)~u^{c+it-1}~du~x^{-it}~dt\\
         &=\frac{x^{-c}}{2\pi}\int_{\mathbb{R}}\int_{\mathbb{R}}\phi\left(e^{w}\right)~e^{cw}~e^{itw}~dw~x^{-it}~dt,
     \end{align*}
     where we perform the change of variable $u=e^{w}$. Using Lemma \ref{lemme1}, we obtain 
     \begin{equation}\label{B5}
         \frac{x^{-c}}{2\pi}\int_{\mathbb{R}}\int_{\mathbb{R}}\phi\left(e^{w}\right)~e^{cw}~e^{itw}~dw~x^{-it}~dt=\frac{x^{-c}}{\sqrt{2\pi} }\int_{\mathbb{R}}\mathcal{F}\left(g\right)(t)~e^{-it\log x}~dt,
     \end{equation}
     where $g(v):=e^{cv}\phi(e^v)$. Using the Fourier inversion integral, we deduce that the right-hand side of \eqref{B5} is equal to 
     \begin{equation}
         {x^{-c}}g(\log x)=\phi(x).
     \end{equation}
     This ends the proof of Theorem \ref{Thm1}.
\end{proof}
The method of proof of Theorem \ref{Thm1} does not hold if one considers the general case of the space $X_c$ defined as
\begin{equation}
    X_c:=\left\{f\left|\right.~ \int_{\mathbb{R}^+}x^c \left|f(x)\right|dx~<+\infty~.\right\}
\end{equation}
For the general case, the proof of the inversion Mellin formula is based on approximation theory, in particular the properties of the Mellin-Gauss-Weistrass kernel. For further details, we refer the reader to~\cite{Butzer1997}. In the context of this paper we have for all $c>0$ 
$$\mathcal{S}(\mathbb{R}^+)\subset X_c,$$
which simplifies considerably the proof of Theorem \ref{Thm1}.
\subsection{The Mellin transform of the Schwartz space $\mathcal{S}(\mathbb{R}^+)$}
\indent Before stating the main result of this section, we introduce some definitions that will be used in the sequel.
\begin{definition}(Singular expansion) Let $\Phi(s)$ be meromorphic in some area $\Omega$ with $S$ including all the poles of $\Phi$ in $\Omega$. A singular expansion of $\Phi(s)$ in $\Omega$ is a formal sum of singular elements of $\Phi(s)$ at all points of $S$. We note it as follows:
$$\Phi(s)\asymp E.$$
\end{definition}
\noindent Note that the singular expansion is only a formal sum, which only shows the poles of the function and its singularities.\\
\indent Let us now define the space $\mathcal{M}^+$ of functions with a fundamental strip in $St(0,+\infty)$ and a meromorphic continuation to the strip $St(-\infty,0) $ such that for all $N\in\mathbb{N}^{0}$, a given function $\psi$ in $\mathcal{M}^+$ admits the singular expansion for $s\in St(-N,0)$
$$\psi(s) \asymp \sum_{n=0}^{N-1}\frac{a_n}{s+n}~.$$
\indent We also define the sub-space $\tilde{\mathcal{M}}^+$ as follows 
\begin{multline}\label{Mtilde+}
    \tilde{\mathcal{M}}^+:=\left\{\psi\in\mathcal{M}^+\left|\right.\forall c>0:\psi\left(c+i\cdot\right)\in\mathcal{S}(\mathbb{R});\right.\\\left.~\forall \alpha\in\mathbb N, \forall c\in\mathbb R^-\setminus \mathbb Z^-:\sup_{t\in\mathbb{R}}|t^\alpha\psi(c+it)|<+\infty ~\right\}.
\end{multline}
Our main result is summarized in the following Theorem:
\begin{theorem}\label{mainres}
    The map
    \begin{align}\label{sp263}
        \mathcal{M}:\mathcal{S}(\mathbb{R}^+)&\longrightarrow \tilde{\mathcal{M}}^+\nonumber\\
        \phi&\longmapsto \mathcal{M}(\phi)(c+it):=\int_{\mathbb{R}^+}\phi(x) ~x^{c+it}~dx
    \end{align}
    is an isomorphism with the following inverse map 
    \begin{align}\label{B16}
        \mathcal{M}^{-1}:\tilde{\mathcal{M}}^+&\longrightarrow \mathcal{S}(\mathbb{R}^+)\nonumber\\
        \psi&\longmapsto \left\{
\begin{array}{rl}
  \mathcal{M}^{-1}(\psi)(x)&=\frac{x^{-c}}{2\pi }\int_{\mathbb{R}}\psi\left(c+it\right)x^{-it}dt,~~\forall c>0~\mathrm{and~}x>0\\
\partial^{k}\mathcal{M}^{-1}(\psi)(0) &= k!~a_k,~~~\forall k\in\mathbb{N},
\end{array}
\right.
    \end{align}
where $a_k$ is the residue of $\psi$ at $s=-k$.
\end{theorem}
\begin{proof}
\eqref{sp263} is well-defined by proposition \ref{prop1}. Now, we need to verify that the map \eqref{B16} is also well-defined, that is for $\psi$ in $\tilde{\mathcal{M}}^+$ we have 
\begin{equation}\label{BE265}
    \mathcal{M}^{-1}(\psi)\in\mathcal{S}(\mathbb{R}^+).
\end{equation}
\begin{itemize}
\item First, we establish that $\mathcal{M}^{-1}(\psi)$ is infinitely differentiable on $\mathbb{R}^+$. We start by studying the behaviour of the function 
\begin{equation}
   \mathcal{M}^{-1}(\psi)(x)=\frac{x^{-c}}{2\pi }\int_{\mathbb{R}}\psi\left(c+it\right)x^{-it}dt
\end{equation}
as $x\rightarrow 0$. Given $N\in\mathbb{N}^*$, the set of poles in $St(-N,+\infty)$ will be denoted in the sequel as ${S}_N$. Given $\alpha>0$, we define the rectangular contour $\mathcal{C}(T)$ by the segments 
\begin{equation}
    \left[\alpha-iT,\alpha+iT\right],~~~\left[\alpha+iT,\gamma+iT\right],~~~\left[\gamma+iT,\gamma-iT\right],~~~\left[\gamma-iT,\alpha-iT\right]
\end{equation}
with $\gamma:=-N-\frac{1}{2}$. Let us consider the integral
\begin{equation}
    \mathcal{I}\left(T\right):=\frac{1}{2\pi i}\int_{\mathcal{C}(T)}\psi(s)~x^{-s}ds~.
\end{equation}
For\footnote{$T$ must be larger than $\mathrm{Im}(s)$ for all $s\in S_N$. In this case, $\mathrm{Im}(s)=0$} $T>0$, remembering that $\psi\in\tilde{\mathcal{M}}^+$, we obtain by Cauchy's theorem 
\begin{equation}\label{B6}
    \mathcal{I}(T)=\sum_{n=0}^{N-1} a_n~\mathrm{Res}\left(\frac{x^{-s}}{s+n}\right)=\sum_{n=0}^{N-1} a_n x^n.
\end{equation}
Using Theorem \ref{Thm1} together with the fact that $\psi$ is holomorphic in $St(0,+\infty)$, we obtain for $\alpha>0$ and $x>0$
\begin{equation}\label{B7}
   \lim_{T\rightarrow +\infty} \frac{1}{2\pi i}\int_{\alpha-iT}^{\alpha+ iT}\psi(\alpha+it)~x^{-\alpha-it}dt= \mathcal{M}^{-1}(\psi)(x)~.
\end{equation}
Now, we consider the term 
\begin{equation}
    \int_{\alpha\pm iT}^{\gamma\pm iT}\psi(s)x^{-s}ds.
\end{equation}
For $0<x<\eta$ with $\eta>0$ , we have 
\begin{equation}\label{B9}
    \left|\int_{\alpha+iT}^{\gamma+iT}\psi(s)x^{-s}ds\right|\leq  \sup_{\substack{c\in(\gamma,\alpha)}}\left|\psi\left(c+iT\right)\right|~\frac{\left|x^{-\gamma}-x^{-\alpha}\right|}{|\log x|}~.
\end{equation}
For $T\in\mathbb{R}^*$, $\psi$ is continuous w.r.t. the variable $c\in\mathbb{R}$ and this implies for fixed $T$ in $\mathbb{R}^*$ that 
\begin{equation}
   \exists \tilde{c}\in [\gamma,\alpha]:~~~~\left|\psi\left(\tilde{c}+iT\right)\right|:=\sup_{\substack{c\in(\gamma,\alpha)}}\left|\psi\left(c+iT\right)\right|
\end{equation}
By proposition \ref{prop1} and lemma \ref{lemme1}, we deduce that 
\begin{equation}
    \lim_{T\rightarrow +\infty}\psi\left(\tilde{c}+iT\right)=0,
\end{equation}
which combined with \eqref{B9} gives  
\begin{equation}\label{B274}
    \lim_{T\rightarrow +\infty}\int_{\alpha+iT}^{\gamma+iT}\left|\psi(s)\right|x^{-s}ds=0~,~~~\forall 0<x<\eta~.
\end{equation}
Without loss of generality, we choose $\eta=1$. Similar arguments lead to 
\begin{equation}\label{B275}
    \lim_{T\rightarrow +\infty}\int_{\alpha-iT}^{\gamma-iT}\left|\psi(s)\right|x^{-s}ds=0~,~~~\forall 0<x<1~.
\end{equation}
Combining (\ref{B6}), (\ref{B274}) and \eqref{B275} we deduce for all $0<x<1$ and $N\in \mathbb{N}^*$ that 
\begin{equation}\label{BE276}
\mathcal{M}^{-1}(\psi)(x)=\sum_{n=0}^{N-1}a_n x^n-\frac{1}{2\pi }\int_{\mathbb{R}}\psi\left(\gamma+it\right)x^{-\gamma-it}dt.
\end{equation}
Using that 
\begin{equation}
    \left|\psi\left(\gamma+it\right)x^{-it}\right|\leq \left|\psi\left(\gamma+it\right)\right|
\end{equation}
together with 
\begin{equation}\label{BE277}
\int_{\mathbb{R}}|t|^{\beta}\left|\psi\left(\gamma+it\right)\right|dt\leq 2\pi\sup_{\substack{0\leq \alpha\leq 2\\t\in\mathbb {R}}} \left|t^{\alpha+\beta}\psi(\gamma+it)\right|<+\infty,
\end{equation}
we deduce that 
\begin{equation}
    \int_{\mathbb{R}}\psi\left(\gamma+it\right)x^{-it}dt
\end{equation}
is infinitely differentiable differentiable w.r.t. $x$. Hence, we obtain\footnote{Remember that $\gamma=-N-\frac{1}{2}$.} for all $N\geq 2$ and $0<x<1$ 
\begin{equation}\label{BE279}
    \partial_x^{N-1} \mathcal{M}^{-1}(\psi)(x)=(N-1)! a_{N-1}-x^{\frac{3}{2}}\int_{\mathbb{R}}\prod_{k=0}^{N-2}\left(N+\frac{1}{2}-it-k\right)\psi\left(\gamma+it\right)x^{-it}dt.
\end{equation}
This establishes that $\mathcal{M}^{-1}(\psi)$ is infinitely differentiable for $x\in(0,1)$. Now let us verify that it is also the case for $x=0$. Again using the bound \eqref{BE277}, we deduce that the second term on the right-hand side of \eqref{BE279} vanishes if $x\rightarrow 0$ which implies that 
\begin{equation}\label{B1}
    \lim_{x\rightarrow 0}\partial_x^{N-1} \mathcal{M}^{-1}(\psi)(x)=(N-1)! a_{N-1}=\partial_x^{N-1} \mathcal{M}^{-1}(\psi)(0),~~~\forall N\geq 2.
\end{equation}
Using \eqref{BE276}, we deduce that  
\begin{equation}\label{B2}
    \lim_{x\rightarrow 0} \mathcal{M}^{-1}(\psi)(x)=a_{0}= \mathcal{M}^{-1}(\psi)(0).
\end{equation}
\eqref{B1} together with \eqref{B2} prove that $\mathcal{M}^{-1}(\psi)$ is infinitely differentiable at $x=0$.  \\
\item For $x>0$, we have 
\begin{equation}\label{B11}
    \phi(x)=\frac{x^{-c}}{2\pi }\int_{\mathbb{R}}dt~\psi(c+it)x^{-it}.
\end{equation}
(\ref{B11}) does not depend on $c>0$ according to proposition 5 of~\cite{Butzer1997} with the fundamental strip $St(0,\infty)$. Hence, we write 
\begin{equation}
    \phi(x)=\frac{x^{-p}}{2\pi }\int_{\mathbb{R}}dt~\psi\left(p+it\right)x^{-it},~~~\forall p\in\mathbb{N}^*.
\end{equation}
We can prove by induction using the theorem of the derivation inside the integral that the $N^{\mathrm{th}}$-order derivative of $\phi$ exists and is given for $x\geq1$ and $p\in\mathbb{N}^*$ by  
\begin{equation}\label{B13}
    \partial_x^N\mathcal{M}^{-1}(\phi)(x)=\frac{(-1)^{N}}{2\pi}\int_{\mathbb{R}}\prod_{l=0}^{N-1}(it+p+l)~\psi(p+it)~x^{-it-p-N}~dt.
\end{equation}
Remembering that $\psi\in\tilde{\mathcal{M}}^+$, we have that \eqref{B13} is well-defined for $x\geq 1$, $p\in\mathbb{N}^*$ and $N\in\mathbb{N}^*$. This together with  \eqref{BE279} implies that $\mathcal{M}^{-1}(\phi)$ is in $\mathcal{C}^{\infty}(\mathbb{R}^+)$. 
\item Now, we need to verify that for all $(p,N)\in\mathbb{N}^2$
\begin{equation}
    x^p~\partial^N_x\mathcal{M}^{-1}(\phi)(x)
\end{equation}
is uniformly bounded with respect to $x$ in $\mathbb{R}^+$. For $0<x<1$, we use \eqref{BE279} to obtain that 
\begin{equation}\label{BE288}
    \left|x^p\partial^N_x\mathcal{M}^{-1}(\phi)(x)\right|\leq O(1)~\left(a_{N}+\sum_{\alpha=0}^{N-1}\mathcal{N}_{\alpha}(\psi)\right),
\end{equation}
where $\mathcal{N}_{\alpha}(\psi):=\sup_{t\in\mathbb{R}}\left|t^{\alpha}\psi\left(\gamma+it\right)\right|$. For $x\geq 1$, we use \eqref{B13} to deduce that 
\begin{equation}\label{BE289}
    \left|x^p\partial^N_x\mathcal{M}^{-1}(\phi)(x)\right|\leq O(1)~\sum_{\alpha=0}^{N-1}\mathcal{N}_{\alpha}(\psi).
\end{equation}
$O(1)$ denotes in both \eqref{BE288}-\eqref{BE289} a constant that only depends on $p$ and $N$. This ends the proof of \eqref{BE265}
\end{itemize}
The last point we need to verify to end the proof of Theorem \ref{MellinThm} is that the maps $\mathcal M$ and $\mathcal M^{-1}$ satisfy
\be\label{MM}
\mathcal{M}^{-1}\mathcal{M}\phi=\phi\qquad\text{and}\qquad \mathcal{M}\mathcal{M}^{-1}\psi=\psi\,,
\ee
for $\phi\in\mathcal S(\mathbb R^+)$ and $\psi\in\tilde{\mathcal M}^+$. Let $\phi$ in $\mathcal{S}(\mathbb{R}^+)$, we have by the definition of the map $\mathcal{M}^-1$
\begin{equation}
\begin{array}{rl}\label{BE291}
\mathcal{M}^{-1}\left(\mathcal{M}(\phi)\right)(x)&=\frac{x^{-c}}{2\pi }\int_{\mathbb{R}}\mathcal{M}(\phi)\left(c+it\right)x^{-it}dt,~~\forall c>0~\mathrm{and~}x>0\\
\partial^{k}\mathcal{M}^{-1}(\mathcal{M}(\phi))(0) &= \phi^{(k)}(0),~~~\forall k\in\mathbb{N}.
\end{array}
\end{equation}
The second line in \eqref{BE291} is a consequence of the fact that $\mathcal{M}(\phi)$ belongs to $\tilde{\mathcal{M}}^+$, which implies that it is meromorphic on the half-plane with simple poles at $s=-k$ and residues $\phi^{(k)}(0)$ for all $k\in\mathbb{N}$. Theorem \ref{Thm1} gives
$$\mathcal M^{-1}\mathcal M \phi=\phi$$
for $\phi\in\mathcal S(\mathbb R^+)$. \\
To show that $\psi\in\tilde{\mathcal{M}}^+$, we have for $\psi\in\tilde{\mathcal M}^+$
\be
\begin{aligned}
\left[\mathcal M\left(\mathcal M^{-1}\psi\right)\right](c+i\lambda)&=\int_0^\infty d\omega ~\omega^{c+i\lambda-1}\int_{\mathbb R}\frac{d\lambda'}{2\pi}\omega^{c'+i\lambda'}\psi(c'+i\lambda')\,,\qquad c,c'>0\\
&=\int_{\mathbb R} \frac{du}{2\pi}~e^{i\lambda u} e^{(c-c')u}\int_{\mathbb R} d\lambda' e^{-i\lambda' u}\psi(c'+i\lambda')\\
&=\int_{\mathbb R} \frac{du}{2\pi}~\mathcal F (\psi)(c'+iu)\, e^{i(\lambda-i(c-c')u)}=\psi(c+i\lambda)\,,
\end{aligned}
\ee
where in the second line we performed the change of variable $\omega=e^u$ and we used the fact that $\psi(c+i\cdot)\in\mathcal S(\mathbb R)$ for $c>0$ to take its Fourier transform.
\end{proof}

%%%%%%%%%%%%%%%%%%%%%%%%%%%%%%%%%%%%%%%%%%%%%%%%%%%%%%%%%%%%%%%%%
\section{The Mellin Transform of Tempered Distributions}\label{sec3}
%%%%%%%%%%%%%%%%%%%%%%%%%%%%%%%%%%%%%%%%%%%%%%%%%%%%%%%%%%%%%%%%%
Given two functions $f,g\in \mathcal S(\mathbb R)$, we define the \textit{generalized convolution product} $f\ast g$ as follows
\be
\left(f\ast g\right)(x) := \int_{\mathbb R} dt f(t) \, g(x+t) e^{-ita}\, .
\ee
We introduce the following lemma which will be used in establishing Parseval's relation:
\begin{lemma}\label{GenConv}(Generalized convolution theorem)
Given the functions $f,g\in \mathcal S(\mathbb R)$ and $\hat f,\hat g$ their respective Fourier transforms, we have 
\be\label{genconv}
\hat{f} (a-s) \hat{g}(s)= \mathcal F\left[ (f\ast g) \right](s)\, .
\ee
\end{lemma}
\begin{proof}
We start from the l.h.s. of~\eqref{genconv} 
\be
\hat{f} (a-s) \hat{g}(s)=\int_{\mathbb R} dt \int_{\mathbb R} du~e^{-it(a-s)}f(t) \, e^{-ius} g(u)\, .
\ee
Performing a change of variable $u=x+t$, we obtain
\be\label{3.444}
\int_{\mathbb R} dt \int_{\mathbb R} dx~e^{-ixs}f(t)g(x+t)e^{-ita}\,.
\ee
Since $f,g\in \mathcal S(\mathbb R)$, we have
\be
\int_{\mathbb R}dt\int_{\mathbb R} dx \left|e^{-ixs}f(t)g(x+t)e^{-ita}\right|=\int_{\mathbb R}dt\left|f(t)\right|\int_{\mathbb R} dx \left|g(x+t)\right|<\infty\,.
\ee
Fubini's theorem together with \eqref{3.444} concludes the proof of lemma~\ref{GenConv}.
\end{proof}
 
\subsection{Parseval's Relation for the Mellin Transform}
From this subsection onward, we change our notations to match the celestial holography literature: $x$ is replaced by $\omega$, $s$ by $\Delta$ and $t$ by $\lambda$. In the new notation, the momenta of massless particles are parameterized as follows
\be\label{Momentum}
p^\mu\equiv\omega q^\mu=\omega\left(\frac{1+x^2}{2},x^a,\frac{1-x^2}{2}\right)\, ,
\ee
where Greek indices run over the $d+2$ bulk spacetime dimensions and the Latin indices run over the $d$ celestial sphere dimensions. $\omega$ designates the energy of the massless particle.
\begin{lemma}\label{Parseval}(Parseval's relation for the Mellin transform of Schwartz functions) Given two functions $f,g\in\mathcal S(\mathbb R^+)$ and their respective Mellin transforms $\tilde f, \tilde g\in\tilde{\mathcal M}^+$, we have \footnote{writing the Lorentz invariant phase space element $\frac{d^3p}{2p^0}$ we get the measure associated to the energy dependence $\omega d\omega$.}
\be\label{parseval}
\int_0^\infty d\omega\, \omega f(\omega)g(\omega)=\int_{c-i\infty}^{c+i\infty}\frac{d\Delta}{2i\pi} ~\tilde f(2-\Delta)\tilde g(\Delta)\, ,
\ee
where $c\in(0,2)$. 
\end{lemma}
\begin{proof}
We start from the r.h.s. of equation~\eqref{parseval} by writing the Mellin transforms $\tilde f$ and $\tilde g$ explicitly
\be
\begin{aligned}
\int_{c-i\infty}^{c+i\infty}\frac{d\Delta}{2i\pi} \tilde f(2-\Delta)\tilde g(\Delta)&=\int_{\mathbb R} \frac{d\lambda}{2\pi}\int_0^\infty d\omega d\omega' \omega^{1-c-i\lambda}\omega'^{c-1+i\lambda} f(\omega)g(\omega')\\
&=\int_{\mathbb R}\frac{d\lambda}{2\pi}\int_{\mathbb R} dx\, dx'\, e^{-i\lambda x}e^{i\lambda x'}F(x) G(x)\, ,
\end{aligned}
\ee
where we performed in the last line the change of variables $\omega=e^x$ and $\omega'=e^{x'}$ and introduced the functions $F(x)=e^{(2-c)x}f(e^x)$ and $G(x)=e^{cx'}g(e^{x'})$. Using lemma~\ref{lemme1}, we have that $F,G\in\mathcal S(\mathbb R)$. Therefore, we obtain
\be
\begin{aligned}
\int_{c-i\infty}^{c+i\infty}\frac{d\Delta}{2i\pi} ~\tilde f(2-\Delta)\tilde g(\Delta)&=\int_{\mathbb R} \frac{d\lambda}{2\pi}~\hat{F}(\lambda)\hat{G}(-\lambda)\\
&=\mathcal F^{-1}\left[\hat{F}(\lambda)\hat{G}(-\lambda)\right](0)\,.
\end{aligned}
\ee
Then using lemma~\ref{GenConv} with $a=0$, we have
\be
\begin{aligned}
\int_{c-i\infty}^{c+i\infty}\frac{d\Delta}{2i\pi}~\tilde f(2-\Delta)\tilde g(\Delta)&= \int_{\mathbb R} dt F(t)G(t)\\
&=\int_{\mathbb R} dt~\,e^{2t} f(e^t)g(e^t)\\
&=\int_0^\infty d\omega \,\omega f(\omega)g(\omega)\, ,
\end{aligned}
\ee
where in the last line we performed the change of variable $\omega=e^t$. This concludes the proof of lemma~\ref{Parseval}.
\end{proof}

We can generalize Parseval's relation to the action of tempered distributions on Schwartz functions.
\begin{theorem}\label{MellinThm}(Mellin transform of tempered distributions)
Consider $f\in \mathcal S(\mathbb R^+)$ and $A\in\mathcal S'(\mathbb R^+)$. If $f$ is the inverse Mellin transform of $\tilde f\in \tilde{\mathcal M}^+$ then we can define the Mellin transform $\tilde A$ of $A$ through
\be
\langle\tilde{A},\tilde{f}\rangle_{{{\tilde{\mathcal M}}^{+\prime}},{\tilde{\mathcal M}}^+}=\langle A,f\rangle_{\mathcal S'(\mathbb R^+),\mathcal S(\mathbb R^+)}\, ,
\ee
where the bracket on the l.h.s. is inspired by a Parseval-type relation
\be\label{MellinBracket}
\langle\tilde A,\tilde f\rangle_{\tilde{\mathcal M}^{+\prime},\tilde{\mathcal M}^+}
=\int_{c-i\infty}^{c+i\infty}\frac{d\Delta}{2i\pi} \,\tilde A(\Delta) \tilde f(2-\Delta)\, ,\qquad 2-c\in St(0,\infty)\, .
\ee
We thus obtain
\be\label{MellinDist}
\tilde A=\int_0^\infty d\omega\, \omega^{\Delta-1} A(\omega).
\ee
\end{theorem}
\begin{proof}
\be
\begin{aligned}
\langle\tilde A,\tilde f\rangle_{\tilde{\mathcal M}^{+\prime},\tilde{\mathcal M}^+}
&=\int_{c-i\infty}^{c+i\infty}\frac{d\Delta}{2i\pi} \,\tilde A(\Delta) \tilde f(2-\Delta)\\
&=\int_0^\infty d\omega \, \omega A(\omega)f(\omega)\\
&=\int_0^\infty d\omega \int_{c'-i\infty}^{c'+i\infty}\frac{d\Delta}{2i\pi}\, \omega^{1-\Delta} A(\omega)\tilde f(\Delta)\,,\quad c'\in(0,\infty)\\
&=\int_0^\infty d\omega \int_{c-i\infty}^{c+i\infty}\frac{d\Delta}{2i\pi}\, \omega^{\Delta-1} A(\omega)\tilde f(2-\Delta)\,,\quad 2-c\in(0,\infty)\, ,
\end{aligned}
\ee
where in the last line we performed the change of variable $\Delta\rightarrow 2-\Delta$. The above expression is equal to equation~\eqref{MellinBracket} by defining $\tilde A$ to be~\eqref{MellinDist}.
\end{proof}

\begin{theorem}(Inverse Mellin transform of elements in $\tilde{\mathcal M}^{+\prime}$)
Equivalently we can define the inverse Mellin transform $A$ of $\tilde A\in\tilde{\mathcal M}^{+\prime}$ through 
\be
\langle A,f\rangle_{\mathcal S'(\mathbb R^+),\mathcal S(\mathbb R^+)}=\langle\tilde{A},\tilde{f}\rangle_{{{\tilde{\mathcal M}}^{+\prime}},{\tilde{\mathcal M}}^+}\, ,
\ee
where $\tilde f\in\tilde{\mathcal M}^{+}$ is the Mellin transform of $f\in S(\mathbb R^+)$. Similarly we obtain
\be
A(\omega)=\int_{c-i\infty}^{c+i\infty}\frac{d\Delta}{2i\pi} \omega^{-\Delta}\tilde A(\Delta)\,,\quad 2-c\in(0,\infty)\, .
\ee
\end{theorem}

%%%%%%%%%%%%%%%%%%%%%%%%%%%%%%%%%%%%%%%%%%%%%%%%%%%%%%%%%%%%%%%%%
\section{Applications: Graviton Celestial Amplitudes}\label{sec4}
%%%%%%%%%%%%%%%%%%%%%%%%%%%%%%%%%%%%%%%%%%%%%%%%%%%%%%%%%%%%%%%%%
In this section, we compute rigorously tree-level graviton celestial amplitudes by taking the Mellin transform of momentum-space amplitudes, which are tempered distributions. In particular, we focus on results found in~\cite{Puhm:2019zbl} where the graviton celestial amplitudes considered in $d+2=4$ bulk dimensions. Our method is based on results from section~\ref{sec3}.
%%%%%%%%%%%%%%%%%%%%%%%%%%%%%%%%%%%%%%%%%%%%%%%%%%%%%%%%%%%%%%%%%%
\subsection{Three-Graviton Celestial Amplitude}
%%%%%%%%%%%%%%%%%%%%%%%%%%%%%%%%%%%%%%%%%%%%%%%%%%%%%%%%%%%%%%%%%
The 3-point MHV tree-level graviton amplitude is 
\be\label{MomentumBasis}
\Acal_{--+}(\omega_i,z_i,\bz_i)= \left(\frac{\omega_1 \omega_2}{\omega_3}\frac{z_{12}^3}{z_{23}z_{31}}\right)^2\delta^{(4)}\left(\sum_{i=1}^3\epsilon_i \omega_i q_i\right)\, ,
\ee
where $z_i=x_i^1+i x_i^2$, $z_{ij}=z_i-z_j$ and $\epsilon_i=\pm 1$ if the particle is outgoing/incoming. Using the parameterization~\eqref{Momentum},  the momentum conserving $\delta$ can be written as
\be\label{DeltaMomentum}
\delta^{(4)}\left(\sum_{i=1}^3\epsilon_i \omega_i q_i\right)=\frac{1}{4\omega_3^2}\frac{sgn(z_{23}z_{31})}{z_{23}z_{31}}\delta\left(\omega_1-\alpha_1 \omega_3\right)\delta\left(\omega_2-\alpha_2 \omega_3\right)\delta\left(\bz_{23}\right)\delta\left(\bz_{31}\right)\, ,
\ee
where $\alpha_1=\frac{\epsilon_3}{\epsilon_1}\frac{z_{23}}{z_{12}}$ and $\alpha_2=\frac{\epsilon_3}{\epsilon_2}\frac{z_{31}}{z_{12}}$ and $\bz_{ij}=z_{ij}^*$. Performing the integrals (Mellin transforms) on all $\omega_i$'s, we obtain~\cite{Puhm:2019zbl}
\be\label{BoostBasis}
\tilde\Acal_{--+}(\Delta_i,z_i,\bz_i)=\frac{z_{12}^6}{|z_{23}z_{31}|^3}\delta(\bz_{23})\delta(\bz_{31})\alpha_1^{\Delta_1+1}\alpha_2^{\Delta_2+1}\Theta(\alpha_1)\Theta(\alpha_2)\int_0^\infty d\omega_3 \, \omega_3^{i\Lambda+3c-3}\, ,
\ee
where $\Delta_i=c+i\lambda_i$ and $\Lambda=\sum_{i=1}^3\lambda_i \,\in \mathbb R$. The integral in $\omega_3$ is divergent at first glance. However this divergence is distributional. Indeed the scattering amplitude is a tempered distribution and should be treated as such when taking its Mellin transform.
\par
Since the $(z,\bz)$ dependence of the amplitude goes through when going from momentum to boost basis, we will only consider the energy dependence of the amplitude~\eqref{MomentumBasis} with~\eqref{DeltaMomentum} as a tempered distribution, i.e. acting on functions in $\Scal^3(\mathbb R^+)$. We denote this energy dependence by
\be
\Omega_3(\{\omega_i\},\alpha_1,\alpha_2)=\frac{(\omega_1\omega_2)^2}{\omega_3^4}\delta(\omega_1-\alpha_1\omega_3)\delta(\omega_2-\alpha_2\omega_3)\, .
\ee
We can then write the bracket of duality between $\Omega_3$ and functions $f_i\in\mathcal S^3(\mathbb R^+)$
\be\label{S'S}
\begin{aligned}
\langle\Omega_3,\{f_{i}\}\rangle_{\Scal',\Scal} &=\prod_{i=1}^3\left(\int_{0}^\infty  d\omega_i \omega_i\right) \frac{(\omega_1\omega_2)^2}{\omega_3^4}\delta(\omega_1-\alpha_1\omega_3)\delta(\omega_2-\alpha_2\omega_3)\, f_1(\omega_1)f_2(\omega_2)f_3(\omega_3)\\
&=\alpha_1^3\alpha_2^3\int_0^\infty d\omega_3 \omega_3^3 f_1(\alpha_1\omega_3)f_2(\alpha_2 \omega_3)f_3(\omega_3)\, .
\end{aligned}
\ee
The above integral is bounded since the functions $f_i$ are in $\Scal(\mathbb R^+)$.\\
In the boost basis, we know from theorem~\ref{MellinThm} that the Mellin transform of $\Omega$ is the part of~\eqref{BoostBasis} that depends on the conformal dimensions $\Delta_i$
\be
D_3(\{\Delta_i\},\alpha_1,\alpha_2)=\prod_{i=1}^3\left(\int_0^\infty d\omega_i \,\omega_i^{\Delta_i-1}\right)\Omega_3(\{\omega_i\},\alpha_1,\alpha_2)=\alpha_1^{\Delta_1+1}\alpha_2^{\Delta_2+1}\int_0^\infty d\omega_3 \omega_3^{i\Lambda +3c-3}\, .
\ee
Using 
\begin{align}\label{In}
I_n=&\int_{\mathbb{R}} du~e^{n u} e^{i\Lambda u} =\int_{\mathbb{R}}du \,\sum_{j=0}^\infty \frac{(nu)^j}{j!} e^{i\Lambda u}\nonumber\\
=&\int_{\mathbb{R}} du\,\sum_{j=0}^\infty \frac{(-in)^j}{j!}\p_\Lambda^j e^{i\Lambda u}=2\pi\sum_{j=0}^\infty \frac{(-in)^j}{j!}\p_\Lambda^j\delta(\Lambda)\, ,  
\end{align}
where $n=3c-2$, we can write its bracket with three functions $\tilde f_i\in{{\tilde{\mathcal M}}^+}$ with $i\in\{1,2,3\}$
\be\label{Df}
\begin{aligned}
D_f\equiv \langle D_3,\{\tilde f_i\}\rangle_{{{\tilde{\mathcal M}}^{+\prime}},{{\tilde{\mathcal M}}^+}}=\prod_{i=1}^3&\left(\int_{c-i\infty}^{c+i\infty} \frac{d\Delta_i}{2i\pi}\right)\alpha_1^{\Delta_1+1}\alpha_2^{\Delta_2+1}\\
&\times 2\pi\sum_{j=0}^\infty \frac{\left(-i(3c-2)\right)^j}{j!}~\delta^{(j)}(\Lambda) \tilde f_1(2-\Delta_1)\tilde f_2(2-\Delta_2)\tilde f_3(2-\Delta_3)\, ,
\end{aligned}
\ee
where the bracket is deduced from Parseval's relation. 
After performing the change of variable $\Delta_3\to\Delta=\sum_{i=1}^3\Delta_i$ we obtain
\be
\begin{aligned}
D_f=\int_{c-i\infty}^{c+i\infty} &\frac{d\Delta_1}{2i\pi}\frac{d\Delta_2}{2i\pi}\int_{3c-i\infty}^{3c+i\infty} \frac{d\Delta}{2i\pi}
\alpha_1^{\Delta_1+1}\alpha_2^{\Delta_2+1}\\
&\times 2\pi\sum_{j=0}^\infty \frac{\left(-i(3c-2)\right)^j}{j!}~\delta^{(j)}(\Lambda) \tilde f_1(2-\Delta_1)\tilde f_2(2-\Delta_2)\tilde f_3(2+\Delta_1+\Delta_2-\Delta)\, .
\end{aligned}
\ee
Rewriting the $\Delta$ integral in terms of $\Lambda=\Im{(\Delta)}$ 
\be
\int_{3c-i\infty}^{3c+i\infty} \frac{d\Delta}{2i\pi} ~h(\Delta)=\int_{\mathbb{R}} \frac{d\Lambda}{2\pi}~h(3c+i\Lambda)\, ,
\ee
and integrating by parts in the $\Lambda$ variable yields the following result
\be\label{E'E}
D_f=\int_{c-i\infty}^{c+i\infty} \frac{d\Delta_1}{2i\pi}\frac{d\Delta_2}{2i\pi}
\alpha_1^{\Delta_1+1}\alpha_2^{\Delta_2+1}\tilde f_1(2-\Delta_1)\tilde f_2(2-\Delta_2)\tilde f_3(\Delta_1+\Delta_2)\, .
\ee
In \eqref{E'E}, the infinite sum shifts $\Lambda$ by $i(3c-2)$ in the argument of $\tilde f_3$. The integral is indeed well-defined since the functions $\tilde f_i (c+i \cdot)\in \mathcal S(\mathbb R)$.
\par
Now, we check that the bracket $D_f$ is indeed equal to the bracket~\eqref{S'S} in momentum space. From this we can confirm that $D_3$ is indeed the Mellin transform of $\Omega_3$. We start by writing the $\tilde f_i$'s as the Mellin transforms of the functions $f_i$ for $i=1,2,3$
\be
\tilde f_i(\Delta_i)=\int_{\mathbb R^+} d\omega_i~\omega_i^{\Delta_i-1}f_i(\omega_i)\, .
\ee
Then, we have
\be
D_f=\alpha_1^2 \alpha_2^2 \int_{c-i\infty}^{c+i\infty} \frac{d\Delta_1}{2i\pi}\frac{d\Delta_2}{2i\pi}\int_{(\mathbb{R}^+)^3} \prod_{i=1}^3 d\omega_i \left(\frac{\alpha_1\omega_3}{\omega_1}\right)^{\Delta_1-1}\left(\frac{\alpha_2\omega_3}{\omega_2}\right)^{\Delta_2-1}\omega_3~ f_1(\omega_1)f_2(\omega_2)f_3(\omega_3)\,.
\ee
Performing the change of variables $u=\log(\omega_3)$, $v=\log(\omega_1)$ and $w=\log(\omega_2)$, we deduce
\be
\begin{aligned}
D_f=(\alpha_1\alpha_2)^{c+1} \int_{\mathbb R} \frac{d\lambda_1}{2\pi}\frac{d\lambda_2}{2\pi} &e^{i(\lambda_1 \log\alpha_1+\lambda_2 \log\alpha_2)}\\
&\times \int_{\mathbb R^3} du\, dv\, dw\, e^{iu(\lambda_1+\lambda_2)}e^{-iv\lambda_1}e^{-iw\lambda_2} g_1(v) \, g_2(w)\, g_3(u)\, ,
\end{aligned}
\ee
with $g_3(u)=e^{2cu}f_3(e^u)$, $g_1(v)=e^{(2-c)v}f_1(e^v)$ and $g_2(w)=e^{(2-c)w}f_2(e^w)$. Using~\eqref{2dstrip} and lemma~\ref{lemme1}, we have that $g_i\in\mathcal S(\mathbb R)$ for $i\in\{1,2,3\}$. The integrals over $u$, $v$ and $w$ are then the Fourier transforms of the respective functions $g_1$, $g_2$ and $g_3$. Then we obtain
\be
D_f=(\alpha_1\alpha_2)^{c+1}\int_{\mathbb R} \frac{d\lambda_1}{2\pi}\frac{d\lambda_2}{2\pi} e^{i(\lambda_1 \log \alpha_1+\lambda_2\log\alpha_2)}\hat g_3(-\lambda_1-\lambda_2)\hat g_1(\lambda_1)\hat g_2(\lambda_2)\, .
\ee
Using lemma~\ref{GenConv} we have
\be
D_f=(\alpha_1\alpha_2)^{c+1}\int_{\mathbb R} \frac{d\lambda_2}{2\pi} \int_{\mathbb R} dt\, g_3(t) g_1(\log\alpha_1+t)e^{it\lambda_2}e^{i\lambda_2\log\alpha_2}\hat g_2(\lambda_2)\, .
\ee
Remembering that $g_1$, $g_3$ and $\hat g_2$ are in $\mathcal S (\mathbb R)$ and using Fubini's theorem, we obtain
\be
\begin{aligned}
D_f&=(\alpha_1\alpha_2)^{c+1}\int_{\mathbb R} dt g_3(t)g_1(\log\alpha_1+t)g_2(\log\alpha_2+t)\\
&=\alpha_1^3\alpha_2^3\int_{\mathbb R} dt e^{4t}f_3(e^t)f_1(\alpha_1 e^t)f_2(\alpha_2 e^t)\\
&=\alpha_1^3\alpha_2^3\int_0^\infty d\omega \omega^3 f_3(\omega)f_1(\alpha_1 \omega)f_2(\alpha_2 \omega)\, ,
\end{aligned}
\ee
where in the last line we performed the change of variable $\omega=e^t$ which indeed confirms that 
\be
D_f=\langle\Omega_3,\{f_{i}\}\rangle_{\Scal',\Scal}\, ,
\ee
and also that $D_3$ is the Mellin transform of $\Omega_3$.

\subsection{Four-Graviton Celestial Amplitude}
%%%%%%%%%%%%%%%%%%%%%%%%%%%%%%%%%%%%%%%%%%%%%%%%%%%%%%%%%%%%%%%%%
We proceed similarly with the 4-graviton amplitude. Its energy dependence can be expressed as follows~\cite{Puhm:2019zbl}
\be\label{4GravEnergy}
\Omega_4=\frac{\omega_1^3 \omega_2^3}{\omega_3^3 \omega_4^2} \delta(\omega_1-\beta_1 \omega_3) \delta(\omega_2-\beta_2 \omega_3) \delta(\omega_4-\beta_4 \omega_3)  \, ,
\ee
where $\beta_i$ are functions of $z_{ij}$ and $\epsilon_i$. Its action on functions in $\Scal(\mathbb R^+)$ is
\be
\begin{aligned}
\langle\Omega_4,\{f_{i}\}\rangle_{\Scal',\Scal}&=\prod_{i=1}^4\left(\int_0^\infty d\omega_i \omega_i f_i(\omega_i)\right) \Omega_4(\{\omega_i\})\\
&=\frac{\beta_1^4 \beta_2^4}{\beta_4} \int_0^\infty d\omega_3\, \omega_3^5 \,f_1(\beta_1\omega_3)f_2(\beta_2\omega_3)f_3(\omega_3)f_4(\beta_4\omega_3)\, .
\end{aligned}
\ee
Taking the Mellin transform on each $\omega_i$ in~\eqref{4GravEnergy}, we obtain
\be
D_4=\beta_1^{\Delta_1+2}\beta_2^{\Delta_2+2}\beta_4^{\Delta_4-3}\int_0^\infty d\omega_3 \, \omega_3^{\sum_{i=1}^4\Delta_i-3}\, ,
\ee
where we can take $\sum_{i=1}^4\Delta_i=4c+i\Lambda$ with $\Lambda=\sum_{i=1}^4 \lambda_i$. Using equation~\eqref{In} as well as Parseval's formula, we obtain the following duality
\be
\begin{aligned}
\langle D_4,\{\tilde f_i\}\rangle_{{{\tilde{\mathcal M}}^{+\prime}},{{\tilde{\mathcal M}}^+}}=\prod_{\substack{i=1\\i\neq 3}}^4&\left(\int_{c-i\infty}^{c+i\infty}\frac{d\Delta_i}{2i\pi}\right)\,\beta_1^{\Delta_1+2}\beta_2^{\Delta_2+2}\beta_4^{\Delta_4-3}\\
&\times \tilde f_1(2-\Delta_1)\tilde f_2(2-\Delta_2)\tilde f_4(2-\Delta_4)\tilde f_3(\Delta_1+\Delta_2+\Delta_4)\, .
\end{aligned}
\ee
The steps are similar to that of the 3-graviton case. Upon a similar check, we also have that
\be
\langle D_4,\{\tilde f_i\}\rangle_{{{\tilde{\mathcal M}}^{+\prime}},{{\tilde{\mathcal M}}^+}}=\langle\Omega_4,\{f_{i}\}\rangle_{\Scal',\Scal}\, .
\ee

%%%%%%%%%%%%%%%%%%%%%%%%%%%%%%%%%%%%%%
\subsection*{The Existence of an Overlapping Holomorphic Strip}
%%%%%%%%%%%%%%%%%%%%%%%%%%%%%%%%%%%%%%
We saw that the functions $\tilde f$ on which the celestial amplitudes act are holomorphic on the complex half-plane $\mathbb C_+$ and have simple poles at $\Delta=-n$ where $n\in\mathbb N$. However in our expressions, namely~\eqref{E'E}, we need to make sure that there is a holomorphic strip denoted by $S_c$ common to all the functions appearing in the integrand and that encloses a contour along which we can integrate.\\
If the function $\tilde f(\Delta)$ only has simple poles at $\Delta=-n$ then $\tilde f(2-\Delta)$ has poles at $\Delta=2+n$ (c.f. figure~\ref{Sc2}).

\begin{figure}[H]
\centering
\begin{tikzpicture}
\draw[thick,->]
(-3,0)--(3,0);
\draw[thick,->]
(0,-2)--(0,2);
\filldraw[thick,red] (1,0) circle (1.5pt);
\filldraw[thick,red] (1.5,0) circle (1.5pt);
\filldraw[thick,red] (2,0) circle (1.5pt);
\filldraw[thick,red] (2.5,0) circle (1.5pt);
%%%%%%%%%%
\filldraw[thick,blue] (0,0) circle (1.5pt);
\filldraw[thick,blue] (-0.5,0) circle (1.5pt);
\filldraw[thick,blue] (-1,0) circle (1.5pt);
\filldraw[thick,blue] (-1.5,0) circle (1.5pt);
\filldraw[thick,blue] (-2,0) circle (1.5pt);
\filldraw[thick,blue] (-2.5,0) circle (1.5pt);
\fill[green, pattern=north west lines] (0.1,-2) rectangle (0.9,2);
\end{tikzpicture}
\caption{This figure represents the complex $\Delta$-plane. The blue dots are the poles of the functions $\tilde f(\Delta)$ and the red dots are the poles of $\tilde f(2-\Delta)$. The striped region represents the strip of holomorphy $S_c$ common to both functions.}
\label{Sc2}
\end{figure}
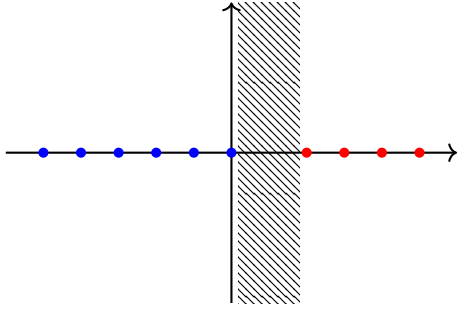
\noindent It is clear that the common strip of holomorphy is
\be\label{2dstrip}
S_c\supset\{\Delta\in \mathbb C\, |\, \Re(\Delta)\in(0,2)\}\, ,
\ee
and that for $c=1$, the contour of integration in~\eqref{E'E} lies within $S_c$. This is in accordance with the celestial holography literature where this choice was made.

\subsubsection*{General Dimension $d+2$}
Notice that the holomorphic strip included the contour where $c=1$ for $d=2$ because the energy integral appearing had an additional factor of $\omega$. This comes from the Jacobian of the coordinate change ($p^\mu$ to $\omega$, $z$ and $\bz$). For general $d$, the factor that would appear is proportional to $\omega^{d-1}$ and that shifts the argument of the second function in~\eqref{Parseval} so that we have $\tilde g(d-\Delta)$. Notice that for $d=2$ we have $\tilde g(2-\Delta)$ which appears in~\eqref{Df}. This increases the width of the common strip of holomorphy such that
\be
\{\Delta\in \mathbb C\, |\, \Re(\Delta)\in(0,d)\}\subset S_c\, .
\ee
In that case the putative choice $c=\frac{d}{2}$ for the contour of integration remains inside of the common strip of holomorphy $S_c$.

\section{Conclusion}\label{sec5}
In this paper, we treated the Mellin transform of distributions in $\mathcal S'(\mathbb R^+)$ in order to have properly well-defined celestial amplitudes. Since tempered distributions are dual to the Schwartz space, we started by characterizing the Mellin transform of $\mathcal S(\mathbb R^+)$ which we denote by $\tilde {\mathcal M}^+$. It is a class of meromorphic functions with simple poles at non-positive integer values of its argument with certain fall-offs in the direction of the imaginary axis. Then, we constructed distributions in $\tilde {\mathcal M}^{+\prime}$ by duality through a Parseval-type relation. Celestial amplitudes belong to the dual space $\tilde {\mathcal M}^{+\prime}$ and we applied the obtained results to the computation of tree-level three- and four-point graviton scattering amplitudes. 
\par
Our framework is only concerned with massless scattering since massive amplitudes require a more complicated transform involving bulk to boundary propagators. Therefore it would be interesting but more difficult to address the same questions for massive scattering amplitudes.\\
Furthermore, we have characterized the space $\tilde {\mathcal M}^{+}$ but refrained from studying its topology. We keep this question for future work.
\\
Finally, the Schwartz space is dense in the Hilbert space $L^2$ which is of interest in QFT and which should ideally define scattering amplitudes. A natural interest in this case is to study the Mellin transform of the larger space $L^2(\mathbb{R}^+)$.

\section*{Acknowledgements}
We would like to thank Guillaume Bossard and Andrea Puhm for useful discussions. YP is supported by the PhD track fellowship of Ecole Polytechnique.

 %%%%%%%%%%%%%%%%%%%%%%%%%%%%%%%%%%%%%%
\appendix

\bibliographystyle{utphys}
\bibliography{Ref}

\providecommand{\href}[2]{#2}\begingroup\raggedright\begin{thebibliography}{10}

\bibitem{Peskin:1995ev}
M.~E. Peskin and D.~V. Schroeder, {\em {An Introduction to quantum field
  theory}}.
\newblock Addison-Wesley, Reading, USA, 1995.

\bibitem{0f5719a4-0389-38e1-93d8-123291531606}
I.~HALPERIN and L.~SCHWARTZ, {\em Introduction to the Theory of Distributions}.
\newblock University of Toronto Press, 1952.
\newblock \url{http://www.jstor.org/stable/10.3138/j.ctt1vxmd4v}.

\bibitem{Pasterski:2016qvg}
S.~Pasterski, S.-H. Shao, and A.~Strominger, ``{Flat Space Amplitudes and
  Conformal Symmetry of the Celestial Sphere},''
  \href{http://dx.doi.org/10.1103/PhysRevD.96.065026}{{\em Phys. Rev. D}
  {\bfseries 96} no.~6, (2017) 065026},
  \href{http://arxiv.org/abs/1701.00049}{{\ttfamily arXiv:1701.00049
  [hep-th]}}.

\bibitem{Pasterski:2017kqt}
S.~Pasterski and S.-H. Shao, ``{Conformal basis for flat space amplitudes},''
  \href{http://dx.doi.org/10.1103/PhysRevD.96.065022}{{\em Phys. Rev. D}
  {\bfseries 96} no.~6, (2017) 065022},
  \href{http://arxiv.org/abs/1705.01027}{{\ttfamily arXiv:1705.01027
  [hep-th]}}.

\bibitem{He:2014laa}
T.~He, V.~Lysov, P.~Mitra, and A.~Strominger, ``{BMS supertranslations and
  Weinberg's soft graviton theorem},''
  \href{http://dx.doi.org/10.1007/JHEP05(2015)151}{{\em JHEP} {\bfseries 05}
  (2015) 151},
\href{http://arxiv.org/abs/1401.7026}{{\ttfamily arXiv:1401.7026 [hep-th]}}.
%%CITATION = ARXIV:1401.7026;%%.

\bibitem{Kapec:2014opa}
D.~Kapec, V.~Lysov, S.~Pasterski, and A.~Strominger, ``{Semiclassical Virasoro
  symmetry of the quantum gravity $ \mathcal{S}$-matrix},''
  \href{http://dx.doi.org/10.1007/JHEP08(2014)058}{{\em JHEP} {\bfseries 08}
  (2014) 058},
\href{http://arxiv.org/abs/1406.3312}{{\ttfamily arXiv:1406.3312 [hep-th]}}.
%%CITATION = ARXIV:1406.3312;%%.

\bibitem{Lysov:2014csa}
V.~Lysov, S.~Pasterski, and A.~Strominger, ``{Low's Subleading Soft Theorem as
  a Symmetry of QED},''
  \href{http://dx.doi.org/10.1103/PhysRevLett.113.111601}{{\em Phys. Rev.
  Lett.} {\bfseries 113} no.~11, (2014) 111601},
\href{http://arxiv.org/abs/1407.3814}{{\ttfamily arXiv:1407.3814 [hep-th]}}.
%%CITATION = ARXIV:1407.3814;%%.

\bibitem{Campiglia:2014yka}
M.~Campiglia and A.~Laddha, ``{Asymptotic symmetries and subleading soft
  graviton theorem},'' \href{http://dx.doi.org/10.1103/PhysRevD.90.124028}{{\em
  Phys. Rev.} {\bfseries D90} no.~12, (2014) 124028},
\href{http://arxiv.org/abs/1408.2228}{{\ttfamily arXiv:1408.2228 [hep-th]}}.
%%CITATION = ARXIV:1408.2228;%%.

\bibitem{He:2014cra}
T.~He, P.~Mitra, A.~P. Porfyriadis, and A.~Strominger, ``{New Symmetries of
  Massless QED},'' \href{http://dx.doi.org/10.1007/JHEP10(2014)112}{{\em JHEP}
  {\bfseries 10} (2014) 112},
\href{http://arxiv.org/abs/1407.3789}{{\ttfamily arXiv:1407.3789 [hep-th]}}.
%%CITATION = ARXIV:1407.3789;%%.

\bibitem{Kapec:2015vwa}
D.~Kapec, V.~Lysov, S.~Pasterski, and A.~Strominger, ``{Higher-dimensional
  supertranslations and Weinberg\textquoteright{}s soft graviton theorem},''
  \href{http://dx.doi.org/10.4310/AMSA.2017.v2.n1.a2}{{\em Ann. Math. Sci.
  Appl.} {\bfseries 02} (2017) 69--94},
  \href{http://arxiv.org/abs/1502.07644}{{\ttfamily arXiv:1502.07644 [gr-qc]}}.

\bibitem{Campiglia:2015yka}
M.~Campiglia and A.~Laddha, ``{New symmetries for the Gravitational
  S-matrix},'' \href{http://dx.doi.org/10.1007/JHEP04(2015)076}{{\em JHEP}
  {\bfseries 04} (2015) 076},
\href{http://arxiv.org/abs/1502.02318}{{\ttfamily arXiv:1502.02318 [hep-th]}}.
%%CITATION = ARXIV:1502.02318;%%.

\bibitem{Campiglia:2015qka}
M.~Campiglia and A.~Laddha, ``{Asymptotic symmetries of QED and Weinberg's soft
  photon theorem},'' \href{http://dx.doi.org/10.1007/JHEP07(2015)115}{{\em
  JHEP} {\bfseries 07} (2015) 115},
\href{http://arxiv.org/abs/1505.05346}{{\ttfamily arXiv:1505.05346 [hep-th]}}.
%%CITATION = ARXIV:1505.05346;%%.

\bibitem{Kapec:2015ena}
D.~Kapec, M.~Pate, and A.~Strominger, ``{New Symmetries of QED},''
\href{http://arxiv.org/abs/1506.02906}{{\ttfamily arXiv:1506.02906 [hep-th]}}.
%%CITATION = ARXIV:1506.02906;%%.

\bibitem{Campiglia:2015kxa}
M.~Campiglia and A.~Laddha, ``{Asymptotic symmetries of gravity and soft
  theorems for massive particles},''
  \href{http://dx.doi.org/10.1007/JHEP12(2015)094}{{\em JHEP} {\bfseries 12}
  (2015) 094}, \href{http://arxiv.org/abs/1509.01406}{{\ttfamily
  arXiv:1509.01406 [hep-th]}}.

\bibitem{Campiglia:2016jdj}
M.~Campiglia and A.~Laddha, ``{Sub-subleading soft gravitons: New symmetries of
  quantum gravity?},''
  \href{http://dx.doi.org/10.1016/j.physletb.2016.11.046}{{\em Phys. Lett.}
  {\bfseries B764} (2017) 218--221},
\href{http://arxiv.org/abs/1605.09094}{{\ttfamily arXiv:1605.09094 [gr-qc]}}.
%%CITATION = ARXIV:1605.09094;%%.

\bibitem{Campiglia:2016hvg}
M.~Campiglia and A.~Laddha, ``{Subleading soft photons and large gauge
  transformations},'' \href{http://dx.doi.org/10.1007/JHEP11(2016)012}{{\em
  JHEP} {\bfseries 11} (2016) 012},
\href{http://arxiv.org/abs/1605.09677}{{\ttfamily arXiv:1605.09677 [hep-th]}}.
%%CITATION = ARXIV:1605.09677;%%.

\bibitem{Strominger:2017zoo}
A.~Strominger, {\em {Lectures on the Infrared Structure of Gravity and Gauge
  Theory}}.
\newblock 3, 2017.
\newblock \href{http://arxiv.org/abs/1703.05448}{{\ttfamily arXiv:1703.05448
  [hep-th]}}.

\bibitem{Pasterski:2017ylz}
S.~Pasterski, S.-H. Shao, and A.~Strominger, ``{Gluon Amplitudes as 2d
  Conformal Correlators},''
  \href{http://dx.doi.org/10.1103/PhysRevD.96.085006}{{\em Phys. Rev. D}
  {\bfseries 96} no.~8, (2017) 085006},
  \href{http://arxiv.org/abs/1706.03917}{{\ttfamily arXiv:1706.03917
  [hep-th]}}.

\bibitem{Schreiber:2017jsr}
A.~Schreiber, A.~Volovich, and M.~Zlotnikov, ``{Tree-level gluon amplitudes on
  the celestial sphere},''
  \href{http://dx.doi.org/10.1016/j.physletb.2018.04.010}{{\em Phys. Lett. B}
  {\bfseries 781} (2018) 349--357},
  \href{http://arxiv.org/abs/1711.08435}{{\ttfamily arXiv:1711.08435
  [hep-th]}}.

\bibitem{Puhm:2019zbl}
A.~Puhm, ``{Conformally Soft Theorem in Gravity},''
  \href{http://dx.doi.org/10.1007/JHEP09(2020)130}{{\em JHEP} {\bfseries 09}
  (2020) 130}, \href{http://arxiv.org/abs/1905.09799}{{\ttfamily
  arXiv:1905.09799 [hep-th]}}.

\bibitem{Albayrak:2020saa}
S.~Albayrak, C.~Chowdhury, and S.~Kharel, ``{On loop celestial amplitudes for
  gauge theory and gravity},''
  \href{http://dx.doi.org/10.1103/PhysRevD.102.126020}{{\em Phys. Rev. D}
  {\bfseries 102} (2020) 126020},
  \href{http://arxiv.org/abs/2007.09338}{{\ttfamily arXiv:2007.09338
  [hep-th]}}.

\bibitem{Arkani-Hamed:2020gyp}
N.~Arkani-Hamed, M.~Pate, A.-M. Raclariu, and A.~Strominger, ``{Celestial
  amplitudes from UV to IR},''
  \href{http://dx.doi.org/10.1007/JHEP08(2021)062}{{\em JHEP} {\bfseries 08}
  (2021) 062}, \href{http://arxiv.org/abs/2012.04208}{{\ttfamily
  arXiv:2012.04208 [hep-th]}}.

\bibitem{Gonzalez:2020tpi}
H.~A. Gonz\'alez, A.~Puhm, and F.~Rojas, ``{Loop corrections to celestial
  amplitudes},'' \href{http://dx.doi.org/10.1103/PhysRevD.102.126027}{{\em
  Phys. Rev. D} {\bfseries 102} no.~12, (2020) 126027},
  \href{http://arxiv.org/abs/2009.07290}{{\ttfamily arXiv:2009.07290
  [hep-th]}}.

\bibitem{Stieberger:2018edy}
S.~Stieberger and T.~R. Taylor, ``{Strings on Celestial Sphere},''
  \href{http://dx.doi.org/10.1016/j.nuclphysb.2018.08.019}{{\em Nucl. Phys. B}
  {\bfseries 935} (2018) 388--411},
  \href{http://arxiv.org/abs/1806.05688}{{\ttfamily arXiv:1806.05688
  [hep-th]}}.

\bibitem{Donnay:2023kvm}
L.~Donnay, G.~Giribet, H.~Gonz\'alez, A.~Puhm, and F.~Rojas, ``{Celestial open
  strings at one-loop},'' \href{http://dx.doi.org/10.1007/JHEP10(2023)047}{{\em
  JHEP} {\bfseries 10} (2023) 047},
  \href{http://arxiv.org/abs/2307.03551}{{\ttfamily arXiv:2307.03551
  [hep-th]}}.

\bibitem{Banerjee:2017jeg}
N.~Banerjee, S.~Banerjee, S.~Atul~Bhatkar, and S.~Jain, ``{Conformal Structure
  of Massless Scalar Amplitudes Beyond Tree level},''
  \href{http://dx.doi.org/10.1007/JHEP04(2018)039}{{\em JHEP} {\bfseries 04}
  (2018) 039}, \href{http://arxiv.org/abs/1711.06690}{{\ttfamily
  arXiv:1711.06690 [hep-th]}}.

\bibitem{Fung}
F.~Kang, ``Generalized mellin transforms,'' \href{http://dx.doi.org/I}{{\em Sc.
  Sinica} {\bfseries 7} (1958) 582--605}.

\bibitem{Zemanian}
A.H.Zemanian, ``{Generalized integral transforms},''
\newblock Dover Publications, New York, 1987.

\bibitem{alma}
O.~P. Misra, {\em Transform analysis of generalized functions}.
\newblock North-Holland mathematics studies ; 119. North-Holland, Amsterdam,
  1986.

\bibitem{Pott}
J.~Potthof, ``On differential operators in white noise analysis,''
  \href{http://dx.doi.org/10.1023/A:1010779905082}{{\em Acta Applicandae
  Mathematicae} {\bfseries 63} (2000) 333--347}.

\bibitem{ZagierAppendixTM}
D.~Zagier, ``The mellin transform and related analytic techniques,'' in {\em
  Appendix to: Zeidler, ``Quantum Field Theory: Basics in Mathematics and
  Physics. A Bridge Between Mathematicians and Physicists}, pp.~305--323.
\newblock Springer, Berlin, 2006.

\bibitem{Butzer1997}
P.~L. Butzer and S.~Jansche, ``A direct approach to the mellin transform.''
  {\em The journal of Fourier analysis and applications [[Elektronische
  Ressource]]} {\bfseries 3} no.~4, (1997) 325--376.
  \url{http://eudml.org/doc/59511}.

\end{thebibliography}\endgroup

\end{document}